\documentclass{CSML}

\def\dOi{11(4:2)2015}
\lmcsheading%
{\dOi}
{1--38}
{}
{}
{Sep.~30, 2014}
{Oct.~15, 2015}
{}

\ACMCCS{[{\bf Theory of computation}]: Models of computation; Logic;
  Formal languages and automata theory}

\usepackage{multicol}
\usepackage{synttree}
\usepackage{hyperref}

\graphicspath{{./}}%helpful if your graphic files are in another directory

\title[Logic and branching automata]{Logic and branching automata\rsuper*}

\author[N.~Bedon]{Nicolas Bedon}
\address{LITIS (EA 4108) -- Universit\'e de Rouen -- France}
\email{Nicolas.Bedon@univ-rouen.fr}

\titlecomment{{\lsuper*}This paper is a long version, with full proofs, of~\cite{DBLP:conf/mfcs/Bedon13}}
\usepackage{diagrams}
\usepackage{epic,eepic,latexsym,amsmath,amssymb,xspace,ifthen}
\usepackage[pst-pdf=md5]{gastex}% compile twice when pictures have changed.

\newcommand{\true}{\text{true}}
\newcommand{\false}{\text{false}}

\newcommand{\Pf}[1]{\ensuremath{{\mathcal P}(#1)}}

\newcommand{\quotientparallel}[2]{\ensuremath{#1\backslash\kern-3pt\backslash#2}}
\begin{document}

\begin{abstract}
In this paper we study the logical aspects of \emph{branching automata}, as defined by Lodaya and Weil.
We first prove that the class of languages of finite N-free posets recognized by branching automata is closed under complementation.
Then we define a logic, named P-MSO as it is a extension of monadic second-order logic with Presburger arithmetic, and show that it is precisely as expressive as branching automata.
As a consequence of the effectiveness of the construction of one formalism from the other, the P-MSO theory of the class of all finite N-free posets is decidable. 
\end{abstract}

\maketitle

\section{Introduction}
\label{sec:intro}

Sequential programs can naturally be modeled with Kleene automata, or equivalently with rational expressions, finite monoids, or monadic second-order (MSO) logic.
The algebraic approach of automata provides an huge toolbox for the study of properties of programs, and has been widely used as a base for a lot of algorithms that manipulate logic formulas. The links between Kleene automata, rational expressions, finite monoids and MSO have important consequences in a lot of domains of computer science and mathematics, some of them are concretely applied as for example in program verification, others are more theoretical as for example in set theory.

Introducing commutativity allows access to models of programs with permutation of instructions, or to concurrent programming. Among the formal tools for the study of commutativity in programs, let us mention for example Mazurkiewicz's traces, integer vector automata or commutative monoids. In this paper we focus on the notion of \emph{branching automata} introduced by Lodaya and Weil~\cite{lodaya98kleene,LW98:Algebra,LW00:sp,lodaya01kleene}. Branching automata are an extension of Kleene automata with particular transitions that naturally model parallelism. Traces of paths in branching automata are not (totally ordered) words as in Kleene automata, but partially ordered sets (posets) of letters, which are structured as traces of programs using the fork-join primitive for concurrency. Those particular posets, called \emph{N-free}, are widely used in the study of concurrency. The fork-join primitive splits an execution flow $f$ into $f_1,\dots,f_n$ concurrent execution flows and joins $f_1,\dots,f_n$ before it continues. Divide-and-conquer concurrent programming naturally uses this fork-join principle. Lodaya and Weil generalized several important results of the theory of Kleene automata to branching automata, for example, a notion of rational expression with the same expressivity as branching automata. They also investigated the question of the algebraic counterpart of branching automata: the sp-algebras are sets equipped with two different associative products, one of them being also commutative. Contrary to the theory of Kleene automata, branching automata do not coincide any more with finite sp-algebras. 

An interesting particular case is the bounded-width rational languages~\cite{LW00:sp}, where the cardinality of the antichains of the posets of languages are bounded by an integer $n$. They correspond to fork-join models of concurrent programs with $n$ as the upper bound of the number of execution flows ($n$ is the number of physical processors). Bounded-width rational languages have a natural characterization in rational expressions, branching automata, and sp-algebras. Taking into account those characterizations, the expressiveness of branching automata corresponds exactly to the finite sp-algebras. Furthermore, Kuske~\cite{Kus00:SPInf} proved that in this case, branching automata coincide also with monadic second-order logic, as it is the case for the rational languages of finite words. As in the general case monadic-second order logic is less expressive than branching automata, the question of an equivalent logic was left open.

In this paper we investigate the question of defining a logic equivalent to branching automata in the general case. 

This paper contains two results regarding branching automata:
\begin{itemize}
\item we prove that the class of languages defined by branching automata is closed under complementation;
\item we define a logic, named \emph{P-MSO} as it is basically monadic second-order logic enriched with Presburger arithmetic, that is exactly as expressive as branching automata.
\end{itemize}

\noindent The (effective) proof of the first result essentially relies on the closure under complementation of rational languages of a finitely generated commutative monoid (Theorem~\ref{th:schutz}, by Sch\"utzenberger and Eilenberg).
The proof of the second result relies on effective constructions from one formalism to the other.
As a consequence, the P-MSO theory of the class of finite N-free posets is decidable.

The paper is organized as follows. Section~\ref{sec:not} recalls basic definitions on posets. Section~\ref{sec:languages} is devoted to branching automata, rational expressions and sp-algebras. The complementation of rational languages is discussed in Section~\ref{sec:complementation}. Section~\ref{sec:PMSO} is devoted to the logical approach of branching automata. Finally, we present an alternative definition for branching automata in Section~\ref{sec:PBranching}.

\section{Notation and basic definitions}
\label{sec:not}

Let $E$ be a set. We denote by ${\mathcal P}(E)$, ${\mathcal P}^+(E)$ and ${\mathcal M}^{>1}(E)$ respectively the set of subsets of $E$, the set of non-empty subsets of $E$ and the set of multi-subsets of $E$ with at least two elements. 
For any integer $n$, the set $\{1,\dots,n\}$ is denoted $[n]$ and the group of permutations of $[n]$ by $S_n$.
The cardinality of $E$ is denoted by $\vert E\vert$.
We use the same notation for sets and multi-sets. We sometimes denote by $e^k$ the multiplicity $k$ of an element $e$ of a multi-set. 

A \emph{poset} $(P,<_P)$ is composed of a set $P$ equipped with a partial ordering $<_P$.
In this paper we consider only finite posets.
For simplicity, by \emph{poset} we always mean \emph{finite} poset.
A \emph{chain} of length $n$ in $P$ is a sequence $p_1<_P\dots<_Pp_n$ of elements of $P$.
An \emph{antichain} $E$ in $P$ is a set of elements of $P$ mutually incomparable for $<_P$.
The \emph{width} of $P$ is the size of a maximal antichain of $P$.
An \emph{alphabet} is a nonempty finite set whose elements are called \emph{letters}.
A poset $(P,<_P,\rho)$ \emph{labeled} by $A$ is composed of a poset $(P,<_P)$ and a map $\rho:P\rightarrow A$ which associates a letter $A$ with any element of $P$. 
Observe that the posets of width 1 labeled by $A$ correspond precisely to the usual finite words: finite totally ordered sequences of letters.
Throughout this paper, we use labeled posets as a generalization of words.
In order to lighten the notation we write $P$ for $(P,<_P,\rho)$ when no confusion is possible.
The unique empty poset is denoted by $\epsilon$.

Let $(P,<_P,\rho_P)$ and $(Q,<_Q,\rho_Q)$ be two disjoint posets labeled respectively by the alphabets $A$ and $A'$.
The \emph{parallel product} of $P$ and $Q$, denoted $P\parallel Q$, is the set $P\cup Q$ equipped with the orderings on $P$ and $Q$ such that the elements of $P$ and $Q$ are incomparable, and labeled by $A\cup A'$ by preservation of the labels from $P$ and $Q$. 
It is defined as $(P\cup Q,<,\rho)$ where $x< y$ if and only if:
\begin{itemize}
\item $x,y\in P$ and $x<_P y$ or
\item $x,y\in Q$ and $x<_Q y$
\end{itemize}
and $\rho(x)=\rho_P(x)$ if $x\in P$, $\rho(x)=\rho_{Q}(x)$ if $x\in Q$.

The \emph{sequential product} of $P$ and $Q$, denoted by $P\cdot Q$ or $PQ$ for simplicity, is the poset $(P\cup Q,<,\rho)$ labeled by $A\cup A'$, such that $x < y$ if and only if one of the following conditions is true:
\begin{itemize}
\item $x\in P$, $y\in P$ and $x<_P y$;
\item $x\in Q$, $y\in Q$ and $x<_Q y$;
\item $x\in P$ and $y\in Q$
\end{itemize}
and $\rho(x)=\rho_P(x)$ if $x\in P$, $\rho(x)=\rho_{Q}(x)$ if $x\in Q$.

Observe that the parallel product is an associative and commutative operation on posets, whereas the sequential product does not commute (but is associative). 
The parallel and sequential products can be generalized to finite sequences of posets. Let $(P_i)_{i\leq n}$ be a sequence of posets. We denote by $\prod_{i\leq n}P_i=P_0\cdot\dots\cdot P_n$ and $\parallel_{i\leq n}P_i=P_0\parallel\dots\parallel P_n$.

The class of \emph{series-parallel} posets, denoted $SP$, is defined as the smallest set containing the posets with zero and one element and closed under finite parallel and sequential product. It is well known that this class corresponds precisely to the class of N-free posets~\cite{Val78,VTL82:SPDigraphs}, in which the exact ordering relation between any four elements $x_1, x_2, x_3, x_4$ cannot be $x_1<x_2$, $x_3<x_2$ and $x_3<x_4$. The class of series-parallel posets labeled by an alphabet $A$ is denoted $SP(A)$. We write $SP^+$ for $SP-\{\epsilon\}$ and $SP^+(A)$ for $SP^+(A)-\{\epsilon\}$.

A poset $P$ has a \emph{sequential (resp. parallel) factorization} if $P=P_1\cdot P_2$ (resp. $P=P_1\parallel P_2$) for some nonempty posets $P_1$ and $P_2$. A sequential factorization $P=P_1\cdot\dots\cdot P_n$ is \emph{maximal} if each $P_i$, $i\in[n]$, has no sequential factorization. The definition of the notion of \emph{maximal parallel factorization} is similar. Posets having a parallel factorization are called \emph{parallel posets}. The \emph{sequential posets} are those of cardinality 1 and those having a sequential factorization.

\section{Rational languages, automata and recognizability}
\label{sec:languages}

A \emph{language} of a set $X$ is a subset of $X$.
Let $A$ be an alphabet.
The sequential and parallel product of labeled posets can naturally be extended to languages of $SP(A)$.
If $L_1, L_2\subseteq SP(A)$, then $L_1 \cdot L_2=\{ P_1 \cdot P_2 : P_1\in L_1 , P_2\in L_2 \}$ and $L_1 \parallel L_2= \{ P_1 \parallel P_2 : P_1\in L_1 , P_2\in L_2 \}$.

\subsection{Rational languages}

Let $A$ and $B$ be two alphabets and let $P\in SP(A)$, $L\subseteq SP(B)$ and $ \xi \in A$.
We define the language $L\circ_\xi P\subseteq SP(A\cup B)$ by substituting non-uniformly in $P$ each element labeled by $\xi$ by a labeled poset of $L$.
This substitution $L \circ_\xi$ is the homomorphism from $(SP(A),\parallel, \cdot)$ into the power-set algebra $({\mathcal P}(SP(A \cup B)),\parallel, \cdot)$ with $a \mapsto \{a\}$ for all $a\in A$, $a\not=\xi$, and $\xi \mapsto L$. 
It can be easily extended from labeled posets to languages of posets.
Using this, we define the substitution and the iterated substitution on languages. By the way the usual Kleene rational operations~\cite{Kle56} are recalled.
Let $L$ and $L'$ be languages of $SP(A)$:
\begin{align*}
  L \circ_\xi L'& = \mathop{\cup}\limits_{ P\in L'}  L \circ_\xi P \\ 
  L^{*\xi} & =\mathop{\cup}\limits_{i\in{\mathbb{N}}} L^{i\xi}\text{ with }L^{0\xi}=\{\xi\}\text{ and }L^{(i+1)\xi}=(\mathop{\cup}\limits_{j\leq i}L^{j\xi})\circ_\xi L \\ 
  L^* & = \{\prod_{i<n}P_i : n\in {\mathbb{N}}, P_i\in L\}\hskip2cm L^+ = \{\prod_{i<n}P_i : 0<n\in {\mathbb{N}}, P_i\in L\}\\\ 
\end{align*}
A language $L \subseteq SP^+(A)$ is
\emph{rational} if it is empty, or obtained from the letters of the alphabet~$A$ using usual rational operators~: finite union $\cup$,
finite concatenation $\cdot$, and finite iteration $^+$, 
and using also the finite parallel product $\parallel$, substitution $\circ_{\xi}$ and iterated substitution $^{*\xi}$, provided that in $L^{*\xi}$ any element labeled by $\xi$ in a labeled poset $P\in L$ is incomparable with another element of $P$. This latter condition excludes from the rational languages those of the form $(a\xi b)^{*\xi}=\{a^n\xi b^n : n\in\mathbb{N}\}$, for example, which are known to be not Kleene rational. Observe also that the usual Kleene rational languages are a particular case of the rational languages defined above, in which the operators $\parallel$, $\circ_{\xi}$ and $^{*\xi}$ are not allowed.

\begin{exa}
  \label{ex:abc}
  Let $A=\{a,b,c\}$ and $L=c\circ_\xi(a\parallel(b\xi))^{*\xi}$. Then $L$ is the smallest language containing $c$ and such that if $x\in L$, then $a\parallel(bx)\in L$:
  $$L=\{c,a\parallel(bc),a\parallel(b(a\parallel(bc))),\dots\}$$
\end{exa}

Let $L$ be a language where the letter $\xi$ is not used.
In order to lighten the notation we use the following abbreviation:
$$L^\circledast=\{\epsilon\}\circ_\xi(L\parallel\xi)^{*\xi}= \{\parallel_{i<n}P_i : n\in {\mathbb{N}}, P_i\in L\}\hskip2cm L^\oplus=L^\circledast-\{\epsilon\}$$
$L^*$ and $L^+$ are the sequential iterations of $L$ whereas $L^\circledast$ and $L^\oplus$ are its parallel iterations.

\subsection{Branching automata}
\label{subsec:automata}

Branching automata are a generalization of usual Kleene automata.
They were introduced by Lodaya and Weil~\cite{lodaya98kleene,LW98:Algebra,LW00:sp}.

A \emph{branching automaton} (or just \emph{automaton} for short) is a tuple $\mathcal{A}=(Q, A, E, I, F)$ where $Q$ is a finite set of
states, $A$ is an alphabet, $I\subseteq Q$ is the set of \emph{initial states},
$F\subseteq Q$ the set of \emph{final states}, and $E$ is the set of \emph{transitions} of ${\mathcal A}$.
The set of transitions of $E$ is partitioned into $E=(E_\text{seq}, E_\text{fork}, E_\text{join})$, according to the different kinds of transitions:
\begin{itemize}
\item $E_\text{seq}\subseteq (Q\times A\times Q)$ contains the \emph{sequential} transitions, which are usual transitions of Kleene automata;
\item $E_\text{fork}\subseteq Q\times {\mathcal M}^{>1}(Q)$ and $E_\text{join}\subseteq{\mathcal M}^{>1}(Q)\times Q$ are respectively the sets of \emph{fork} and \emph{join} transitions.
\end{itemize}
Sequential transitions $(p,a,q) \in Q\times A\times Q$  are sometimes denoted by $p \mathop{\to}\limits^{a}q$.
The \emph{arity} of a fork (resp. join) transition $(q,R)\in Q\times {\mathcal M}^{>1}(Q)$ (resp. $(R,q)\in {\mathcal M}^{>1}(Q)\times Q$) is $\vert R\vert$. 

We now turn to the definition of paths in automata. The definition we use in this paper is different, but equivalent to, the one of Lodaya and Weil~\cite{lodaya98kleene,LW98:Algebra,LW00:sp,lodaya01kleene}.
Paths in automata are posets labeled by transitions.
A \emph{path} $\gamma$ from a state $p$ to a state $q$ is either the empty poset (in this case $p=q$), or a non-empty poset labeled by transitions, with a unique minimum and a unique maximum element. The minimum element of $\gamma$ is mapped either to a sequential transition of the form $(p,a,r)$ for some $a\in A$ and $r\in Q$ or to a fork transition of the form $(p,R)$ for some $R\in{\mathcal M}^{>1}(Q)$. Symmetrically, the maximum element  of $\gamma$ is mapped either to a sequential transition of the form $(r',a,q)$ for some $a\in A$ and $r'\in Q$ or to a join transition of the form $(R',q)$ for some $R'\in{\mathcal M}^{>1}(Q)$. The states $p$ and $q$ are respectively called \emph{source} (or \emph{origin}) and \emph{destination} of $\gamma$. Two paths $\gamma$ and $\gamma'$ are \emph{consecutive} if the destination of $\gamma$ is also the source of $\gamma'$.
Formally, the paths $\gamma$ labeled by $P\in SP^+(A)$ in $\mathcal{A}$ are defined by induction on the structure of~$P$:
\begin{itemize}
\item for any transition $t=(p,a,q)$, then $t$ is a path from $p$ to $q$, labeled by $a$;
\item for any finite set of paths $\{\gamma_0,\dots,\gamma_k\}$ ($k\geq 1$) respectively labeled by $P_0,\dots,P_k\in SP^+(A)$, from $p_0,\dots,p_k$ to $q_0,\dots, q_k$, if $t=(p,\{p_0,\dots,p_k\})$ is a fork transition and $t'=(\{q_0,\dots,q_k\},q)$ a join transition, then $\gamma=t(\parallel_{j\leq k} \gamma_j)t'$ is a path from $p$ to $q$ and labeled by $\parallel_{j\leq k} P_j$;
\item for any non-empty finite sequence $\gamma_0,\dots,\gamma_k$ of consecutive paths respectively labeled by $P_0,\dots,P_k$, then $\prod_{j<k+1}\gamma_j$ is a path labeled by $\prod_{j<k+1}P_j$ from the source of  $\gamma_{0}$ to the destination of $\gamma_{k}$;
\end{itemize}
Observe that non-empty paths are labeled posets of two different forms: $t$ or $tPt'$ for some transitions $t,t'$ and some labeled poset $P$.
In an automaton $\mathcal{A}$, the existence of a path $\gamma$ from $p$ to $q$ labeled by $P\in SP(A)$ is denoted by $\gamma : p \mathop{\Longrightarrow}\limits_{\mathcal{A}}^{P} q$. 
A state $s$ is a \emph{sink} if $s$ is the destination of any path originating in~$s$.

A labeled poset is \emph{accepted} by an automaton if it is the nonempty label of a path, called \emph{successful}, leading from an initial state to a final state. The language $L(\mathcal{A})$ is the set of labeled posets accepted by the automaton $\mathcal{A}$. A language $L$ is \emph{regular} if there exists an automaton $\mathcal{A}$ such that $L=L(\mathcal{A})$.

\begin{exa}
  \label{ex:oneAAtLeast}
  Figure~\ref{fig:ontAAtLeast} represents an automaton on the alphabet $A=\{a,b\}$ that accepts $P\in SP^+(A)$ iff $P$ contains at least one $a$.
  It has one initial state $1$ and one final state $2$, two sequential transitions labeled by $a$ from $1$ to $2$ and from $2$ to $2$, two sequential transitions labeled by $b$ from $1$ to $1$ and from $2$ to $2$, two fork transitions $(1,\{1,1\})$ and $(2,\{2,2\})$, and three join transitions $(\{1,1\},1)$, $(\{2,2\},2)$ and $(\{1,2\},2)$. 
  \begin{figure}[htbp]
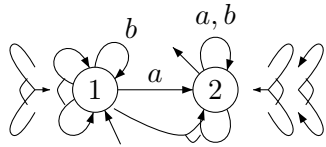

  \begin{center}
    \gasset{Nw=6,Nh=6,loopdiam=4}
    \begin{gpicture}%(20,0)
      \node(1)(0,0){$1$}\imark[iangle=295](1)
      \node(2)(16,0){$2$}\fmark[fangle=135](2)
      \drawedge(1,2){$a$}
      \drawloop[loopangle=60](1){$b$}
      \drawloop[loopangle=90](2){$a,b$}

      \drawloop[loopangle=135](1){}
      {\gasset{AHnb=0,ATnb=1}
      \drawloop[loopangle=235](1){}}
      \drawcurve[AHnb=0](-4,1.5)(-5,0)(-4,-2)

      \drawcurve[AHnb=0](-8,3)(-11,6)(-8,0)
      \drawcurve[AHnb=0](-8,-3)(-11,-6)(-8,0)
      \drawcurve[nb=1](-8,0)(-7,0)(-6,0)
      \drawcurve[AHnb=0](-9,1.5)(-10,0)(-9,-2)

      \drawcurve[AHnb=0](23,3)(26,6)(23,0)
      \drawcurve[AHnb=0](23,-3)(26,-6)(23,0)
      \drawcurve[nb=1](23,0)(22,0)(21,0)
      \drawcurve[AHnb=0](24,1.5)(25,0)(24,-2)

      \drawcurve[AHnb=1](27,0)(30,6)(27,3)
      \drawcurve[AHnb=1](27,0)(30,-6)(27,-3)
      \drawcurve[AHnb=0](28,1.5)(29,0)(28,-2)

      {\gasset{AHnb=1,ATnb=0}
      \drawloop[loopangle=270](2){}}
      \drawcurve[AHnb=0](2.5,-2.5)(12,-6)(14.3,-4)
      \drawcurve[AHnb=0](12,-6)(13,-7)(14,-6)
    \end{gpicture}
  \caption{An automaton that accepts $P\in SP^+(A)$ iff $P$ contains at least one $a$.}
  \label{fig:ontAAtLeast}
  \end{center}
\end{figure}
\end{exa}

\begin{thm}[Lodaya and Weil~\cite{lodaya98kleene}]
  \label{th:KleeneBranching}
  Let $A$ be an alphabet, and $L\subseteq SP^+(A)$.
  Then $L$ is regular if and only if it is rational.
\end{thm}

\begin{exa}
  \label{ex:anbnAuto}
  On its left side, Figure~\ref{fig:autoAB} represents the automaton $$\mathcal{A}=(\{1,2,3,4,5,6\},\{a,b\},E,\{1\},\{6\})$$ where the set of sequential transitions is $E_\text{seq}=\{(2,a,4),(3,b,5)\}$, the set of fork transitions is $E_\text{fork}=\{(1,\{1,1\}),(1,\{2,3\})\}$ and finally the set of join transitions is $E_\text{join}=\{(\{6,6\},6),(\{4,5\},6)\}$.
  On the right side of the Figure is pictured an accepting path labeled by $a\parallel b\parallel a\parallel b$. 
  Actually, $L(\mathcal{A})=(a\parallel b)^{\oplus}$.
  \begin{figure}[htbp]
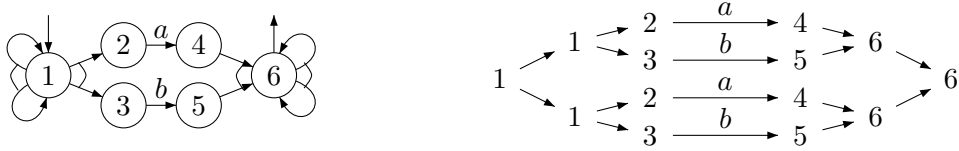

    \begin{center}
      \gasset{Nw=6,Nh=6,loopdiam=4}
      \begin{gpicture}
        \node(1)(-10,-12){$1$}\imark[iangle=90](1)
        \node(2)(0,-8){$2$}
        \node(3)(0,-16){$3$}
        \node(4)(10,-8){$4$}
        \node(5)(10,-16){$5$}
        \node(6)(20,-12){$6$}\fmark[fangle=90](6)
        \drawedge(2,4){$a$}
        \drawedge(3,5){$b$}
        \drawedge(1,2){}
        \drawedge(1,3){}
        \drawloop[loopangle=135](1){}
        {\gasset{AHnb=0,ATnb=1}
          \drawloop[loopangle=235](1){}}
        \drawcurve[AHnb=0](-6,-10.3)(-5,-12)(-6,-13.7)
        \drawcurve[AHnb=0](-14,-10.5)(-15,-12)(-14,-14)
        \drawedge(4,6){}
        \drawedge(5,6){}
        {\gasset{AHnb=0,ATnb=1}\drawloop[loopangle=45](6){}}
        \drawloop[loopangle=-45](6){}
        \drawcurve[AHnb=0](24,-10.3)(25,-12)(24,-13.7)
        \drawcurve[AHnb=0](16,-10.5)(15,-12)(16,-14)
        \node[Nframe=n](7)(50,-12.5){$1$}
        \node[Nframe=n](8)(60,-7.5){$1$}
        \node[Nframe=n](9)(70,-5){$2$}
        \node[Nframe=n](10)(60,-17.5){$1$}
        \drawedge(7,8){}
        \drawedge(7,10){}
        \node[Nframe=n](11)(70,-10){$3$}
        \node[Nframe=n](12)(70,-15){$2$}
        \node[Nframe=n](13)(70,-20){$3$}
        \drawedge(10,12){}
        \drawedge(10,13){}
        \node[Nframe=n](14)(90,-10){$5$}
        \node[Nframe=n](15)(90,-15){$4$}
        \node[Nframe=n](16)(90,-20){$5$}
        \node[Nframe=n](17)(100,-17.5){$6$}
        \drawedge(15,17){}
        \drawedge(16,17){}
        \node[Nframe=n](18)(100,-7.5){$6$}
        \node[Nframe=n](19)(90,-5){$4$}
        \node[Nframe=n](20)(110,-12.5){$6$}
        \drawedge(17,20){}
        \drawedge(18,20){}
        \drawedge(19,18){}
        \drawedge(14,18){}
        \drawedge(8,9){}
        \drawedge(8,11){}
        \drawedge(9,19){$a$}
        \drawedge(13,16){$b$}
        \drawedge(11,14){$b$}
        \drawedge(12,15){$a$}
      \end{gpicture}
      \caption{An automaton $\mathcal{A}$ with $L(\mathcal{A})=(a\parallel b)^{\oplus}$ and an accepting path labeled by $a\parallel b\parallel a\parallel b$.}
      \label{fig:autoAB}
    \end{center}
  \end{figure}
\end{exa}

\begin{exa}
  \label{ex:abcAuto}
  Let $A=\{a,b,c\}$ and $L=c\circ_\xi(a\parallel (b\xi))^{*\xi}$ be the language of Examples~\ref{ex:abc} and~\ref{ex:abcAlgebra}.
  Figure~\ref{fig:autoAXIB} represents an automaton $\mathcal{A}$ such that $L(\mathcal{A})=L$.
  \begin{figure}[htbp]
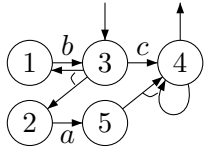

  \begin{center}
    \gasset{Nw=6,Nh=6,loopdiam=4}
    \begin{gpicture}%(20,0)
      \node(2)(0,-8){$1$}
      \node(3)(0,-16){$2$}
      \node(4)(10,-8){$3$}
      \node(5)(20,-8){$4$}
      \node(6)(10,-16){$5$}
      \drawedge(2,4){$b$}
      \drawedge(4,5){$c$}\imark[iangle=90](4)\fmark[fangle=90](5)
      \drawedge[ELside=r](3,6){$a$}
      \drawedge[syo=-1,eyo=-1,ELside=r](4,2){}
      \drawedge(4,3){}
      \drawcurve[AHnb=0](6,-9.3)(5.5,-10)(6,-10.9)

      \drawedge(6,5){}
      {\gasset{AHnb=1,ATnb=0}
      \drawloop[loopangle=260](5){}}
      \drawcurve[AHnb=0](15,-12)(16,-12.5)(17,-12)
   
    \end{gpicture}
  \caption{An automaton $\mathcal{A}$ such that $L(\mathcal{A})=c\circ_\xi(a\parallel (b\xi))^{*\xi}$.}
  \label{fig:autoAXIB}
  \end{center}
\end{figure}
\end{exa}

\begin{prop}[Lodaya and Weil~\cite{LW00:sp}]% Th 4.6 TCSLodayaWeil
  \label{prop:union}
  Let $A$ be an alphabet.
  The class of regular languages of $SP^+(A)$ is closed under finite union.
  Furthermore, if $B$ is an alphabet, $\phi:SP^+(A)\to SP^+(B)$ a morphism of free sp-algebras, and $L$ a regular language of $SP^+(A)$, then $\phi(L)$ is a regular language of $SP^+(B)$.
\end{prop}

\begin{proof}
  The closure under finite union is a direct consequence of the generalization of the notion of cartesian product of automata to branching automata. The closure under direct image by $\phi$ is also an easy generalization of the construction for Kleene rational languages.
\end{proof}

\begin{prop}[Lodaya and Weil~\cite{lodaya01kleene}]% Prop 2.3 lodaya01kleene
  \label{prop:findPath}
  Let $p$ and $q$ be two states of a branching automaton $\mathcal{A}$.
  It is decidable, in polynomial time, if there is a path from $p$ to $q$ in $\mathcal{A}$.
\end{prop}

An automaton is \emph{sequentially separated} if, for all pairs $(p,q)$ of states, all labels of paths from $p$ to $q$ are parallel posets, or all labels of paths from $p$ to $q$ are sequential posets.

The following proposition will be used later in the paper.
\begin{prop}
  \label{prop:seqSepAuto}
  For every automaton $\mathcal{A}$ there is a sequentially separated automaton $\mathcal{B}$ such that $L(\mathcal{A})=L(\mathcal{B})$.
\end{prop}

\begin{proof}
  Let $\mathcal{A}=(Q,A,E,I,F)$.
  The states of $\mathcal{B}$ are the elements of $Q\times\mathbb{B}$.
  In $\mathcal{B}$, there is 
  \begin{itemize}
    \item a sequential transition $((p,b),a,(q,\true))$ iff $(p,a,q)\in E$,
    \item a fork transition $((p_0,b_0),\{(p_1,b_1),\dots,(p_n,b_n)\})$ iff $(p_0,\{p_1,\dots,p_n\})\in E$,
    \item a join transition $(\{(p_1,b_1),\dots,(p_n,b_n)\},(p_0,\false))$ iff $(\{p_1,\dots,p_n\},p_0)\in E$.
  \end{itemize}
  The initial (resp. final states) of $\mathcal{B}$ are those of the form $(p,b)$ with $p\in I$ (resp. $i\in F$).
  Clearly $L(\mathcal{A})=L(\mathcal{B})$ and $\mathcal{B}$ is sequentially separated.
\end{proof}

\subsection{Rationality and semi-linearity}
\label{subsec:ratSemiLin}

A subset $L$ of $A^\circledast$ is \emph{linear} if it has the form $$L=a_1\parallel\dots\parallel a_k\parallel\bigl(\cup_{i\in I}(a_{i,1}\parallel\dots\parallel a_{i,k_i})\bigr)^\circledast$$ where the $a_i$ and $a_{i,j}$ are elements of $A$ and $I$ is a finite set. It is \emph{semi-linear} if it is a finite union of linear sets.

The class of \emph{$\parallel$-rational} languages of $A^\circledast$ is the smallest containing the empty set, $\epsilon$, and closed under finite union, parallel product $\parallel$, and finite parallel iteration $^\circledast$. 
We refer to~\cite{EilSch69} for a proof of the following result:
\begin{thm}
  \label{th:ratSemLin}
  Let $A$ be an alphabet and $L\subseteq A^\circledast$. Then $L$ is $\parallel$-rational if and only if it is semi-linear. Furthermore, the construction from one formalism to the other is effective.
\end{thm}

\begin{prop}
  \label{prop:parallelRational}
  Let $A$ be an alphabet and $L$ be a rational language of $SP^+(A)$.
  Then $L\subseteq A^\circledast$ if and only if $L$ is $\parallel$-rational.
\end{prop}

\begin{proof}
  The inclusion from right to left follows immediately from the definition.
  Let us turn to the inclusion from left to right. There exists an automaton $\mathcal{A}=(Q,A,E,I,F)$ such that $L=L(\mathcal{A})$. For each pair $(p,q)$ of states of $\mathcal{A}$ define $L_{p,q}$ as the set of posets labeling paths from $p$ to $q$. As $P\cdot P'\not\in L$ for all $P,P'\in SP^+(A)$ we can assume that $P\cdot P'\not\in L_{p,q}$ for all $P,P'\in SP^+(A)$ and $p,q\in Q$. Then 
$$L_{p,q}=\mathop\bigcup_{{(p,\{p_1,\dots,p_n\})\in E_\text{fork}\atop (\{q_1,\dots,q_n\},q)\in E_\text{join}}\atop \sigma\in S_n} \parallel L_{p_i,q_{\sigma(i)}}\mathop\bigcup_{(p,a,q)\in E_\text{seq}} \{a\}$$
The set of all such equalities forms a finite system of equations, which can immediately be re-written as a context-free grammar $G$ where the usual concatenation is replaced by $\parallel$, which commutes. By Parikh's Theorem (see~\cite{Parikh:1966:CL:321356.321364}, Theorem~2), each $L_{p,q}$ is a semi-linear set of $A^\circledast$ with $\epsilon\not\in L_{p,q}$, which can be effectively be computed from $G$. As $L=\cup_{(i,f)\in I\times F}L_{i,f}$, and the class of semi-linear sets is closed under finite union, then $L$ is also semi-linear, hence $\parallel$-rational by Theorem~\ref{th:ratSemLin}. 
\end{proof}

The definitions of linearity, semi-linearity, rationality and $\parallel$-rationality, which are given above over free algebras, can naturally be generalized over (non-free) algebras.
 
\section{Complementation of rational languages}
\label{sec:complementation}

The first result of this paper is stated by the following Theorem which claims, together with Proposition~\ref{prop:union}, that the class of rational languages of N-free posets is closed under boolean operations.
\begin{thm}
  \label{th:complement}
  Let $A$ be an alphabet.
  The class of rational languages of $SP^+(A)$ is effectively closed under complement.
\end{thm}

The proof of Theorem~\ref{th:complement} relies on the closure under complementation of the class of $\parallel$-rational languages of commutative monoids (Theorem~\ref{th:schutz}).
\begin{thm}[Eilenberg and Sch\"utzenberger, Theorem~III of~\cite{EilSch69}]
  \label{th:schutz}
  If $X$ and $Y$ are $\parallel$-rational subsets of a commutative monoid $M$, then $X\cap Y$ and $Y-X$ are also $\parallel$-rational subsets of~$M$.
\end{thm}
As emphasized in~\cite{Saka:thAutoFr}, if $M$ is finitely generated then Theorem~\ref{th:schutz} is effective.
Theorem~\ref{th:schutz} was first proved by Ginsburg and Spanier~\cite{GS:AMS64} in the case of finitely generated free commutative monoids. The following proposition is a corollary of Theorem~\ref{th:schutz}:
\begin{prop}[Eilenberg and Sch\"utzenberger, Corollary~III.2 of~\cite{EilSch69}]
  \label{prop:morphInv}
  If $\varphi:M'\to M$ is a morphism of commutative monoids, $M'$ is finitely generated and $X$ is a $\parallel$-rational subset of $M$, then $\varphi^{-1}(X)$ is a $\parallel$-rational subset of~$M'$.
\end{prop}

\medskip

Before going into the details we need to introduce the necessary notions on algebras for languages of N-free posets.
For the basic notions on algebra we refer to~\cite{Alm94}.
An \emph{sp-algebra} $(S,\cdot,\parallel)$ consists of a set $S$ equipped with two operations $\cdot$ and $\parallel$, respectively called \emph{sequential} and \emph{parallel product}, such that $(S,\cdot)$ is a semigroup ($\cdot$ is associative) and $(S,\parallel)$ is a commutative semigroup. Observe that the notion of an sp-algebra equipped with a \emph{neutral element} 1 (verifying $1\cdot x=x\cdot 1=x\parallel 1=x$ for any element $x$ of the sp-algebra) corresponds to bimonoid in~\cite{Bloom199655}. For each alphabet $A$ there exists a free sp-algebra which is isomorphic to $SP(A)$ (and which is also denoted by $SP(A)$). 
For simplicity we often denote an sp-algebra $(S,\cdot,\parallel)$ by $S$.
A morphism $\varphi:S\to T$ between two sp-algebras \emph{recognizes} $X\subseteq S$ if $X=\varphi^{-1}(R)$ for some $R\subseteq T$. Sometimes the reference to $\varphi$ is omitted and we say that $T$ recognizes $X$.
The following propositions are easy generalizations of well-known results on semigroups (see~\cite[Prop.~1.8 and 1.9]{Pin84} for the semigroup versions).
\begin{prop}
  If $\varphi: A\to S$ is a map from an alphabet $A$ to an sp-algebra $S$, there exists a unique morphism $\overline{\varphi}:SP^+(A) \to S$ such that $\varphi(a)=\overline{\varphi}(a)$ for all $a\in A$. Furthermore, $\overline{\varphi}$ is surjective iff $\varphi(A)$ is a generator of $S$.
\end{prop}
\begin{prop}
  Let $A$ be an alphabet, $\varphi:SP^+(A)\to S$ and $\psi:T\to S$ two morphisms of sp-algebras, with $\psi$ surjective. There exists a morphism $\mu:SP^+(A)\to T$ such that $\varphi=\psi\mu$. Furthermore, $\mu$ recognizes any $L\subseteq SP^+(A)$ recognized by $\varphi$.
\end{prop}
A subset $X$ of an sp-algebra $S$ is \emph{recognizable} if there exists a \emph{finite} sp-algebra $T$ and a morphism $\varphi:S\to T$ such that $\varphi$ recognizes $X$. A \emph{congruence} $\sim$ of sp-algebras is an equivalence relation compatible with the operations, ie. $x\sim y$ implies that $(1): u\cdot x\cdot v\sim u\cdot y\cdot v$ and $(2): u\parallel x\parallel v\sim u\parallel y\parallel v$ for all $u,v$. Actually, as the parallel product commutes in sp-algebras, the condition~(2) is equivalent to $x\parallel u\sim y\parallel u$ for all $u$. An equivalence relation has \emph{finite index} if it has a finite number of equivalence classes.
It is well-known that the map $\varphi_\sim:S\to S/\mathord{\sim}$ which associates to any element of $S$ its equivalence class in the quotient sp-algebra $S/\mathord{\sim}$ can be extended in a unique way into a morphism of sp-algebras.
Let $X$ be a set whose elements are called \emph{variables} and $S$ be an sp-algebra. 
A \emph{term} on $S$ is a full binary tree whose leafs are either variables or elements of $S$, and nodes are a sequential or a parallel product.
Formally, the set ${\mathcal T}$ of terms on $S$ is defined inductively by $X\subseteq {\mathcal T}$, each element of $S$ is in ${\mathcal T}$ and, for all $t,t'\in {\mathcal T}$, $t\cdot t'\in {\mathcal T}$ and $t\parallel t'\in {\mathcal T}$. 
Observe that a N-free poset labeled by $A$ can be thought of as a term of $A$ (which may not be unique) without variables, and reciprocally (note that a term corresponds to a unique N-free poset).
A value can be associated to any term $t$ whose leaves are all elements of $S$ by the partial function $e:{\mathcal T}\to S$ inductively defined by $e(s)=s$ for all $s\in S$, $e(t\parallel t')=t\parallel t'$ and $e(t\cdot t')=t\cdot t'$.
Let $L\subseteq S$. The \emph{syntactic congruence} $\sim_L$ of $L$ on $S$ is defined by: for all $x,y\in S$, $x\sim_L y$ if, for any term $t(x_0,\dots,x_n)$ on $S$ and any $s_1,\dots,s_n\in S$, $$e(t(x,s_1,\dots,s_n))\in L\iff e(t(y,s_1,\dots,s_n))\in L$$
It is well-known that the quotient sp-algebra $S/\mathord{\sim_L}$ recognizes $L$. Furthermore, the following property holds on $S/\mathord{\sim_L}$:
\begin{prop}[see~\cite{Alm94} or~\cite{LW00:sp}]
  \label{prop:synt}
  Let $S$ and $T$ be two sp-algebras, $L\subseteq S$ and $\varphi:S\to T$ be an onto morphism.
  Then $L=\varphi^{-1}\varphi(L)$ if and only if, for any $x,y\in S$, $\varphi(x)=\varphi(y)$ implies $x\sim_L y$.
\end{prop}

\begin{exa}
  \label{ex:abcAlgebra}
  Let $A=\{a,b,c\}$ and $L=c\circ_\xi(a\parallel(b\xi))^{*\xi}$ be the language of Example~\ref{ex:abc}.
  Let $(S,\cdot,\parallel)$ be the sp-algebra defined by $S=\{a,b,c,s,0,1\}$, $bc=s$, $a\parallel s=c$, $1$ is the neutral element for both sequential and parallel products, and all other products are mapped to $0$. 
  Then $S$ recognizes $L$.
  Indeed, let $\varphi: SP^+(A)\to S$ be the morphism defined by $\varphi(a)=a$, $\varphi(b)=b$ and $\varphi(c)=c$.
  Then $\varphi(L)=\{c\}$ and $L=\varphi^{-1}(c)$.
  Furthermore $S=SP^+(A)/\mathord{\sim_L}$.
\end{exa}

Lodaya and Weil have proved the following connection between recognizable and rational languages:
\begin{thm}[Lodaya and Weil, Theorem~1 of~\cite{lodaya98kleene}]
  \label{th:RecRat}
  Recognizable languages are rational.
\end{thm}

However, the following example (from~\cite{lodaya98kleene}) shows that in general, rational languages are not recognizable.
\begin{exa}
  \label{ex:anbnAlgebra}
  Let $A=\{a,b\}$ and $L=(a\parallel b)^\oplus$.
  Let $\varphi:SP^+(A)\to\mathbb{Z}\cup\{\bot\}$ the morphism defined by $\varphi(a)=1$, $\varphi(b)=-1$, $xy=\bot$ for all $x,y\in \mathbb{Z}\cup\{\bot\}$.
  Then $L=\varphi^{-1}(0)$, and thus $\varphi$ recognizes $L$.
  Furthermore, $SP^+(A)/\mathord{\sim_L}$ is isomorphic to $\mathbb{Z}\cup\{\bot\}$, thus $L$ is not recognizable as a consequence of Proposition~\ref{prop:synt}.
  Example~\ref{ex:anbnAuto} gives an automaton $\mathcal{A}$ with $L(\mathcal{A})=L$.
\end{exa}

\medskip

Let us return to the proof of Theorem~\ref{th:complement}.
The first step is the construction of an algebra from an automaton.
We need to introduce some new definitions, which are applied in Example~\ref{ex:complement} at the end of this section.

\medskip

Let $L\subseteq SP^+(A)$ and $\mathcal{A}=(Q,A,E,I,F)$ be an automaton such that $L(\mathcal{A})=L$.

For every pair $(p,q)$ of states, define $K_{p,q}$ to be the set of multi-sets of pair of states as follows:
\begin{multline*}
K_{p,q}=\{\{(p_1,q_{\sigma(1)}),\dots,(p_n,q_{\sigma(n)})\} : (p,\{p_1,\dots,p_n\})\in E_\text{fork},\\ (\{q_1,\dots,q_n\},q)\in E_\text{join}, \sigma\in S_n\text{ and } p_i \mathop{\Longrightarrow}^{P_i}\limits_{\mathcal{A}} q_{\sigma(i)}, P_i\in SP^+(A), \text{ for all }i\in[n]\}
$$
\end{multline*}
Define also $\mathcal{F}_{p,q}$ to be the smallest set of multi-sets of pairs of states as follows.
Let
$$\mathcal{F}^0_{p,q}=
\begin{cases}
\{\{(p,q)\}\} &\text{ if there exists } p \mathop{\Longrightarrow}^P\limits_{\mathcal{A}} q, P\in SP^+(A),\\
\emptyset &\text{ otherwise.}
\end{cases}
$$
and 
\begin{multline*}
\mathcal{F}_{p,q}^{i+1}=\mathcal{F}_{p,q}^{i}\cup\{ M-(p_j,q_j)\cup X :
M = \{(p_1,q_1)^{k_1},\dots,(p_n,q_n)^{k_n}\}\in \mathcal{F}_{p,q}^{i},\\
j\in[n], k_j>0, X\in K_{(p_j,q_j)} \}
\end{multline*}
% $$\mathcal{F}_{p,q}^{i+1}=\{ M : M\in \mathcal{F}^i_{p,q}\text{ or }\exists M'\in \mathcal{F}_{p,q}^i, (r,s)\in M', X\in K_{r,s} \text{ such that }M=M'[(r,s)\leftarrow X]\}
% $$
Now let $\mathcal{F}_{p,q}=\cup_{i\in\mathbb{N}}\mathcal{F}^{i}_{p,q}$.
Observe that $\emptyset\not\in\mathcal{F}_{p,q}$.
If $P=P_1\parallel\dots\parallel P_n$ is a N-free poset such that there exists a path $\gamma:p\mathop{\Longrightarrow}\limits_{\mathcal{A}}^{P} q$ which is the parallel composition of $n>0$ paths $\gamma_i:p_i\mathop{\Longrightarrow}\limits_{\mathcal{A}}^{P_i} q_i$, for some $p_i,q_i\in Q$, $P_i\in SP^+(A)$, $i\in[n]$, observe that by construction the multi-set $\{(p_1,q_1),\dots,(p_n,q_n)\}$ belongs to $\mathcal{F}_{p,q}$. Reciprocally, if $\{(p_1,q_1),\dots,(p_n,q_n)\}\in \mathcal{F}_{p,q}$, there exist paths $p_i\mathop{\Longrightarrow}\limits_{\mathcal{A}}^{P_i} q_i$ for some $P_i\in SP^+(A)$ for all $i\in[n]$, that can be composed to build a path $p\mathop{\Longrightarrow}\limits_{\mathcal{A}}^{P_1\parallel\dots\parallel P_n} q$.

Let $E$, $I$ be sets, $X=\{ X_i :i\in I\}$ and $Y$ be respectively a set of multi-sets of elements of $E$ and a multi-set of elements of $E$. Set
$$\quotientparallel{\{X_i:i\in I\}}{Y}=\{X_i-Y : Y\subseteq X_i, i\in I\}$$

When $P\in SP^+(A)$, define
\begin{multline*}
\mathcal{R}(P)=\{\{ (p_1,q_1),\dots,(p_n,q_n) : p_i \mathop{\Longrightarrow}\limits_{\mathcal{A}}^{P_i} q_i \text{ for all }i\in[n]\} : \\ P=P_1\parallel\dots\parallel P_n, P_i\in SP^+(A)\text{ for all }i\in[n]\}
\end{multline*}
Thus $\mathcal{R}(P)$ is the set of all finite multi-sets $\{(p_1,q_1),\dots,(p_n,q_n)\}$ over $Q^2$ such that $P=P_1\parallel\dots\parallel P_n$ and, for all $i\in[n]$, $P_i\in SP^+(A)$ and $p_i \mathop{\Longrightarrow}\limits_{\mathcal{A}}^{P_i} q_i$.

Let $\sim_{\mathcal{A}}$ be the relation defined on $SP^+(A)$ by $P\sim_{\mathcal{A}}P'$ if and only if, for all $p,q\in Q$,
\begin{equation}
\label{eq:cong}
\mathop\bigcup_{x\in \mathcal{R}(P)} \quotientparallel{\mathcal{F}_{p,q}}{x}
=
\mathop\bigcup_{x\in \mathcal{R}(P')} \quotientparallel{\mathcal{F}_{p,q}}{x}
\end{equation}
Equivalently, $P\sim_{\mathcal{A}}P'$ if and only if, for all multi-sets $M$ over $Q^2$, there is some $x\in \mathcal{R}(P)$ such that $x\cup M\in\mathcal{F}_{p,q}$ if and only if there is some $x'\in \mathcal{R}(P')$ such that $x'\cup M\in\mathcal{F}_{p,q}$.
Obviously, $\sim_{\mathcal{A}}$ is an equivalence relation.
Also, for all $p,q\in Q$, $P\in SP^+(A)$, we have $\emptyset\in\mathop\bigcup_{x\in \mathcal{R}(P)} \quotientparallel{\mathcal{F}_{p,q}}{x}$ if and only if $p\mathop{\Longrightarrow}\limits_{\mathcal{A}}^{P} q$. Immediately,
\begin{lem}
  \label{lem:congChemin}
  If $P\sim_\mathcal{A}P'$ then $p\mathop{\Longrightarrow}\limits_{\mathcal{A}}^{P} q$ iff $p\mathop{\Longrightarrow}\limits_{\mathcal{A}}^{P'} q$ for all $p,q\in Q$.
\end{lem}

The following lemma shows in particular that $SP^+(A)/\mathord\sim_\mathcal{A}$ is equipped with a structure of sp-algebra.
\begin{lem}
  \label{lem:cong}
  $\sim_{\mathcal{A}}$ is a congruence of sp-algebra.
\end{lem}

\begin{proof}
  First we prove that if $P\sim_{\mathcal{A}}P'$ then $LPR\sim_{\mathcal{A}}LP'R$ for all $L,R\in SP(A)$.
  Let $r\in\mathop\bigcup_{x\in\mathcal{R}(LPR)} \quotientparallel{\mathcal{F}_{p,q}}{x}$.
  If $L=R=\epsilon$ the conclusion is trivially reached.
  Otherwise, $LPR$ is a sequential poset of $SP^+(A)$.
  Assume $r=\{(r_1,s_1),\dots,(r_k,s_k)\}$. By definition of $r$, there exists a path $\gamma:p' \mathop{\Longrightarrow}\limits_{\mathcal{A}}^{LPR} q'$ such that $r\cup\{(p',q')\}\in \mathcal{F}_{p,q}$. By definition of $\mathcal{F}_{p,q}$ there exists a path $\gamma_i:r_i \mathop{\Longrightarrow}\limits_{\mathcal{A}}^{S_i} s_i$ for some $S_i\in SP^+(A)$ and for all $i\in[k]$, and the paths $(\gamma_i:r_i \mathop{\Longrightarrow}\limits_{\mathcal{A}}^{S_i} s_i)_{i\in[k]}$ and $\gamma:p' \mathop{\Longrightarrow}\limits_{\mathcal{A}}^{LPR} q'$ can be used to compose a path $p \mathop{\Longrightarrow}\limits_{\mathcal{A}}^{LPR\parallel S_1\parallel\dots\parallel S_k} q$. The path $\gamma:p' \mathop{\Longrightarrow}\limits_{\mathcal{A}}^{LPR} q'$ can be decomposed into $\gamma:p' \mathop{\Longrightarrow}\limits_{\mathcal{A}}^{L} t_1 \mathop{\Longrightarrow}\limits_{\mathcal{A}}^{P} t_2 \mathop{\Longrightarrow}\limits_{\mathcal{A}}^{R} q'$ for some $t_1,t_2\in Q$. As $P\sim_{\mathcal{A}}P'$ by Lemma~\ref{lem:congChemin} we also have $t_1 \mathop{\Longrightarrow}\limits_{\mathcal{A}}^{P'} t_2$, thus there exists $\gamma':p' \mathop{\Longrightarrow}\limits_{\mathcal{A}}^{L} t_1 \mathop{\Longrightarrow}\limits_{\mathcal{A}}^{P'} t_2 \mathop{\Longrightarrow}\limits_{\mathcal{A}}^{R} q'$, which can be used in parallel with the paths $(\gamma_i:r_i \mathop{\Longrightarrow}\limits_{\mathcal{A}}^{S_i} s_i)_{i\in[k]}$ to build a path  $p \mathop{\Longrightarrow}\limits_{\mathcal{A}}^{LP'R\parallel S_1\parallel\dots\parallel S_k} q$. Thus, $r\in \mathop\bigcup_{x\in\mathcal{R}(LP'R)} \quotientparallel{\mathcal{F}_{p,q}}{x}$.
  
We now show that if $P\sim_{\mathcal{A}}P'$ then $P\parallel P''\sim_{\mathcal{A}}P'\parallel P''$ for all $P''\in SP(A)$. The case $P''=\epsilon$ is a triviality, so we assume that $P''\not=\epsilon$. Let $r=\{(r_1,s_1),\dots,(r_k,s_k)\}\in \mathop\bigcup_{x\in \mathcal{R}(P\parallel P'')} \quotientparallel{\mathcal{F}_{p,q}}{x}$. There exist a decomposition $P\parallel P''=X_1\parallel\dots\parallel X_n$ of $P\parallel P''$, $(p_i,q_i)_{i\in[n]}$,  $S_i\in SP^+(A)$ for all $i\in[k]$, paths $(p_i\mathop{\Longrightarrow}\limits_{\mathcal{A}}^{X_i} q_i)_{i\in[n]}$  and $(r_i\mathop{\Longrightarrow}\limits_{\mathcal{A}}^{S_i} s_i)_{i\in[k]}$,
 that can be composed to form a path $p  \mathop{\Longrightarrow}\limits_{\mathcal{A}}^{X_1\parallel\dots\parallel X_n\parallel S_1\parallel\dots\parallel S_k} q$. If $X_1\parallel\dots\parallel X_n$ is not a maximal parallel factorization of $P\parallel P''$ then the paths $(p_i\mathop{\Longrightarrow}\limits_{\mathcal{A}}^{X_i} q_i)_{i\in[n]}$ can be replaced by $(p'_i\mathop{\Longrightarrow}\limits_{\mathcal{A}}^{X'_i} q'_i)_{i\in[n']}$ where $X'_1\parallel\dots\parallel X'_{n'}$ is a maximal parallel factorization of $P\parallel P''$, such that the paths $(p'_i\mathop{\Longrightarrow}\limits_{\mathcal{A}}^{X'_i} q'_i)_{i\in[n']}$ and the paths $(r_i\mathop{\Longrightarrow}\limits_{\mathcal{A}}^{S_i} s_i)_{i\in[k]}$ can be composed to form a path $p  \mathop{\Longrightarrow}\limits_{\mathcal{A}}^{X'_1\parallel\dots\parallel X'_{n'}\parallel S_1\parallel\dots\parallel S_k} q$. Since $X'_1\parallel\dots\parallel X'_{n'}$ is a maximal parallel factorization of $P\parallel P''$ there exists a partition $(I,J)$ of $[n']$ such that $P=\parallel_{i\in I} X'_i$ and $P''=\parallel_{j\in J} X'_j$. As $P\sim_{\mathcal{A}}P'$ we have
$$\{(p'_j,q'_j):j\in J\}\cup\{(r_i,s_i) : i\in[k]\}\in (\quotientparallel{\mathcal{F}_{p,q}}{\{(p'_i,q'_i):i\in I\}})\cap (\quotientparallel{\mathcal{F}_{p,q}}{x})$$
for some $x\in\mathcal{R}(P')$, 
so
$r\in \quotientparallel{(\quotientparallel{\mathcal{F}_{p,q}}{x})}{\{(p'_j,q'_j):j\in J\}}$, ie. $r\in \quotientparallel{\mathcal{F}_{p,q}}{y}$ with $y=x\cup \{(p'_j,q'_j):j\in J\}$.
  So $\sim_{\mathcal{A}}$ is a congruence of sp-algebra. 
\end{proof}

Let $\varphi_{\sim_{\mathcal{A}}}:SP^+(A)\to SP^+(A)/\mathord\sim_\mathcal{A}$ the morphism which associates to each poset $P\in SP^+(A)$ its equivalence class in $SP^+(A)/\mathord\sim_\mathcal{A}$. Then $\varphi_{\sim_{\mathcal{A}}}$ recognized $L$, since $L(\mathcal{A})=\varphi^{-1}(X)$ where
$$
X=\{\varphi_{\sim_{\mathcal{A}}}(P) : \emptyset\in\mathop\bigcup_{x\in \mathcal{R}(P)}\quotientparallel{\mathcal{F}_{i,f}}{x}\text{ for some }(i,f)\in I\times F\}
$$
Observe that $X$ may be infinite.

\begin{exa}
  Let $\mathcal{A}$ be the automaton of Example~\ref{ex:abcAuto}. Then $SP^+(A)/\mathord{\sim_{\mathcal{A}}}$ is isomorphic to $SP^+(A)/\mathord{\sim_{L(\mathcal{A})}}$ (see Example~\ref{ex:abcAlgebra} for $SP^+(A)/\mathord{\sim_{L(\mathcal{A})}}$).
\end{exa}

Observe that $\sim_{\mathcal{A}}$ may have an infinite index (take, for example, any automaton of language $(a\parallel b)^\oplus$ - see Example~\ref{ex:anbnAuto}).

\begin{lem}
  \label{lem:finiteIndex}
  The number of equivalence classes for $\sim_{\mathcal{A}}$ containing a sequential poset (of $SP^+(A)$) is finite.
\end{lem}

\begin{proof}
  By contradiction, assume that there exists an infinite sequence $(P_i)_{i\in\mathbb{N}}$ of sequential posets such that for all $i,j\in\mathbb{N}$, if $i\not=j$ then $P_i\not\sim_{\mathcal{A}}P_j$.
  To each equivalence class $[P_i]_{\sim_\mathcal{A}}$  we associate the set $K_{P_i}=\{(p,q)\in Q^2 : p \mathop{\Longrightarrow}\limits_{\mathcal{A}}^{P_{i}} q\}$.
  Clearly there exist a finite number of such sets, so there exist $i,j$, with $i\not=j$ such that $K_{P_i}=K_{P_j}$.
  For all $p,q\in Q$, $S\in SP^+(A)$, let $X_{p,q}(R)=\mathop\bigcup_{x\in \mathcal{R}(S)}\quotientparallel{\mathcal{F}_{p,q}}{x}$.
  As $P_i\not\sim_{\mathcal{A}}P_j$ there exist $p,q\in Q$ and wlog. $r=\{(r_1,s_1),\dots,(r_n,s_n)\}\in X_{p,q}(P_i)-X_{p,q}(P_j)$.
  As $r\in X_{p,q}(P_i)$ there exist paths $\gamma_i:r_i \mathop{\Longrightarrow}\limits_{\mathcal{A}}^{S_i} s_i$ for some $S_i\in SP^+(A)$ and for all $i\in[n]$, and $\gamma:p' \mathop{\Longrightarrow}\limits_{\mathcal{A}}^{P_i} q'$, such that the paths $(\gamma_i)_{i\in[n]}$ and $\gamma$ can be used to compose a path $\delta:p \mathop{\Longrightarrow}\limits_{\mathcal{A}}^{P_i\parallel S_1\parallel\dots\parallel S_n} q$. In $\delta$, $\gamma$ can be replaced by $\gamma':p'\mathop{\Longrightarrow}\limits_{\mathcal{A}}^{P_j} q'$ to form a path $p \mathop{\Longrightarrow}\limits_{\mathcal{A}}^{P_j\parallel S_1\parallel\dots\parallel S_n} q$. Thus $r\in X_{p,q}(P_j)$, which is a contradiction.
\end{proof}

Let $X\subseteq SP^+(A)$. We denote the set of \emph{sequential posets} of $X$ by
\begin{align*}
  Seq(X) &= \{P\in X : P\in A\text{ or }\exists P_1,P_2\in SP^+(A)\text{ such that }P=P_1P_2\}
\end{align*}
Denote also by $L_{p,q}$ the set of non-empty labels of paths from state $p$ to state $q$ in $\mathcal{A}$.
We are now going to prove that $\varphi_{\sim_\mathcal{A}}(L(\mathcal{A}))$ is a $\parallel$-rational language of $(\varphi_{\sim_\mathcal{A}}(Seq(SP^+(A))))^\circledast$. If $\sim_{\mathcal{A}}$ has a finite index this is a triviality, so assume that it has an infinite number of equivalence classes, and recall that $\varphi_{\sim_\mathcal{A}}(Seq(L(\mathcal{A})))$ is finite by Lemma~\ref{lem:finiteIndex}. 

We have $$L_{p,q}=\mathop\bigcup_{X\in \mathcal{F}_{p,q}}\mathop\parallel_{(r,s)\in X} Seq(L_{r,s}) \hskip1cm \text{ and } \hskip1cm \varphi_{\sim_\mathcal{A}}(L_{p,q})=\mathop\bigcup_{X\in \mathcal{F}_{p,q}}\mathop\parallel_{(r,s)\in X} \varphi_{\sim_\mathcal{A}}(Seq(L_{r,s}))$$
So, it suffices to show that $\mathcal{F}_{p,q}$ is a $\parallel$-rational set of $(Q\times Q)^\oplus$ in order to prove that $\varphi_{\sim_\mathcal{A}}(L_{p,q})$ is a $\parallel$-rational set of elements of $\varphi_{\sim_\mathcal{A}}(Seq(SP^+(A)))$, and thus so is $\varphi_{\sim_\mathcal{A}}(L(\mathcal{A}))=\cup_{(i,f)\in I\times F}\varphi_{\sim_\mathcal{A}}(L_{i,f})$.

\begin{lem}
  \label{lem:FqpParallelRat}
  $\mathcal{F}_{p,q}$ is a $\parallel$-rational set of $(Q\times Q)^\oplus$.
\end{lem}

\begin{proof}
  First observe that $\emptyset\not\in\mathcal{F}_{p,q}$.
  Build an automaton $\mathcal{B}$ whose alphabet is $Q\times Q$ as follows.
  Take two copies $Q_1$ and $Q_2$ of the states of $\mathcal{A}$.
  For each fork transition $(r,\{r_1,\dots,r_n\})$ of $\mathcal{A}$, add the same fork transition in $\mathcal{Q}_1$.
  For each join transition $(\{s_1,\dots,s_n\},s)$ of $\mathcal{A}$, add the same join transition in $\mathcal{Q}_2$.
  For each pair of states $(r,s)$ such that there is a non-empty path from $r$ to $s$ in $\mathcal{A}$, add a sequential transition from $r$ in $\mathcal{Q}_1$ to $s$ in $\mathcal{Q}_2$, labeled by $(r,s)$.
  The initial state is $p$ in $\mathcal{Q}_1$ and the final state is $q$ in $\mathcal{Q}_2$.
  The language of $\mathcal{B}$ is $\mathcal{F}_{p,q}$.
  There is no path in $\mathcal{B}$ with a sequential transition or a join transition followed by a sequential transition or by a fork transition.
  As a consequence of Theorem~\ref{th:KleeneBranching} and Proposition~\ref{prop:parallelRational}, $L(\mathcal{B})$ is $\parallel$-rational.
\end{proof}

Observe that the construction given in the proof of Lemma~\ref{lem:FqpParallelRat} is effective as a consequence of Proposition~\ref{prop:findPath}.

Define the equivalence relation $\sim_\emptyset^{SP^+(A)}$ over the elements of $SP^+(A)$ by 
$P\sim_\emptyset^{SP^+(A)} P'$ iff 
$\{(p,q)\in Q^2: p \mathop{\Longrightarrow}\limits_{\mathcal{A}}^{P} q\} = \{(p,q)\in Q^2: p \mathop{\Longrightarrow}\limits_{\mathcal{A}}^{P'} q\}$, or equivalently
 $P\sim_\emptyset^{SP^+(A)} P'$ iff 
 $
 \{(p,q)\in Q^2 : \emptyset\in\mathop\bigcup_{x\in \mathcal{R}(P)}\quotientparallel{\mathcal{F}_{p,q}}{x}\}
 =
 \{(p,q)\in Q^2 : \emptyset\in\mathop\bigcup_{x\in \mathcal{R}(P')}\quotientparallel{\mathcal{F}_{p,q}}{x}\}
 $.
Note that $\sim_\emptyset^{SP^+(A)}$ has finite index, since $Q^2$ is finite.
When $D\in\mathcal{P}(Q^2)$, denote by  $\Delta_D^{SP^+(A)}=\{P\in SP^+(A) : p \mathop{\Longrightarrow}\limits_{\mathcal{A}}^{P} q\text{ iff }(p,q)\in D\}$.

\begin{lem}
  \label{lem:equivEmptyRecLpq}
  Let $\varphi:SP^+(A)\to S$ be a morphism of sp-algebras.
  The following conditions are equivalent:
  \begin{enumerate}
  \item $P\sim_\emptyset^{SP^+(A)} P'$ for all $P,P'\in\varphi^{-1}(s)$, $s\in S$;
  \item $\varphi$ recognizes $L_{p,q}$ for all $(p,q)\in Q^2$.
  \end{enumerate}
\end{lem}

\begin{proof}
  If~(1) is true then $\varphi$ recognizes $L_{p,q}$ for all $(p,q)\in Q^2$ since $$L_{p,q}=\mathop\bigcup_{D\in\mathcal{P}(Q^2)\atop(p,q)\in D}\varphi^{-1}(\varphi(\Delta^{SP^+(A)}_D))$$
  Conversely assume~(2) and, by contradiction, that~(1) is false, ie. there exist $s\in S$, $P,P'\in\varphi^{-1}(s)$ such that $P\in\Delta^{SP^+(A)}_D$, $P'\in\Delta^{SP^+(A)}_{D'}$ with $D\not=D'$. Then $\varphi$ can not recognize $L_{p,q}$ with  $(p,q)\in D$ and $(p,q)\not\in D'$, or the converse.
\end{proof}

Let $\varphi:SP^+(A)\to S$ be a morphism of sp-algebras such that $P\sim_\emptyset^{SP^+(A)} P'$ for all $P,P'\in\varphi^{-1}(s)$, $s\in S$. We define the equivalence relation $\sim_\emptyset^{S}$ over the elements of $S$ by $s\sim_\emptyset^{S} s'$ iff there exist $P\in \varphi(s)$, $P'\in \varphi(s')$ such that $P\sim_\emptyset^{SP^+(A)} P'$.
We have $P\sim_\emptyset^{SP^+(A)} P'$ iff $\varphi(P)\sim_\emptyset^{S} \varphi(P')$. 
If $\varphi$ is surjective then $\sim^S_\emptyset$ has finite index, and each equivalence class of $\sim_\emptyset^S$ can be denoted $\Delta_D^S=\{\varphi(P)\in S : P\in\Delta_D^{SP^+(A)}\}$ for some $D\in\mathcal{P}(Q^2)$. Furthermore, $\varphi(\Delta_D^{SP^+(A)})=\Delta_D^{S}$ and $\varphi^{-1}(\Delta_D^{S})=\Delta_D^{SP^+(A)}$.

\begin{lem}
  \label{lem:simEmptyParallelRational}
  Let $\varphi:SP^+(A)\to S$ be a surjective morphism of sp-algebras such that 
  \begin{itemize}
  \item $\varphi$ recognizes $L_{p,q}$ for all $(p,q)\in Q^2$,
  \item $\varphi(L_{p,q})$ is a $\parallel$-rational of $S$ for all $(p,q)\in Q^2$.
  \end{itemize}
  Each equivalence class $\Delta_D^S$ of $\sim_\emptyset^S$ is a $\parallel$-rational set of $S$.
\end{lem}

\begin{proof}
  We have $\varphi^{-1}(\Delta_D^S)=\Delta_D^{SP^+(A)}=\cap_{(p,q)\in D}L_{p,q}-\cup_{(p,q)\not\in D}L_{p,q}$; as $\varphi(L_{p,q})$ is $\parallel$-rational and $\varphi$ recognizes $L_{p,q}$ for all $(p,q)\in Q^2$, and as by Theorem~\ref{th:schutz} the class of $\parallel$-rational sets is closed under finite boolean operations, $\Delta_D^S=\varphi(\cap_{(p,q)\in D}L_{p,q}-\cup_{(p,q)\not\in D}L_{p,q})=\cap_{(p,q)\in D}\varphi(L_{p,q})-\cup_{(p,q)\not\in D}\varphi(L_{p,q})$ is $\parallel$-rational.
\end{proof}

It is clear that the morphism of sp-algebras $\varphi_{\sim_{\mathcal{A}}}:SP^+(A)\to SP^+(A)/\mathord\sim_\mathcal{A}$ verifies the conditions of Lemma~\ref{lem:simEmptyParallelRational}.

\medskip

When $S$ is an sp-algebra, define also the equivalence relation $\sim_\text{seq}^S$ on the elements of $S$ by $s\sim_\text{seq}^Ss'$ iff $xs=xs'$ and $sx=s'x$ for all $x\in S$.
The relation between $\sim_\emptyset^{SP^+(A)/\mathord\sim_\mathcal{A}}$ and $\sim_\text{seq}^{SP^+(A)/\mathord\sim_\mathcal{A}}$ is given by the following lemma:
\begin{lem}
  \label{lem:simEmptyImpliqueSimSeq}
  If $s\sim_\emptyset^{SP^+(A)/\mathord\sim_\mathcal{A}}s'$ then $s\sim_\text{seq}^{SP^+(A)/\mathord\sim_\mathcal{A}}s'$.
\end{lem}

\begin{proof}
  Let $P\in\varphi_{\sim_{\mathcal{A}}}^{-1}(s)$ and $P'\in\varphi_{\sim_{\mathcal{A}}}^{-1}(s')$.
  We have $P\sim_\emptyset^{SP^+(A)} P'$, then $p \mathop{\Longrightarrow}\limits_{\mathcal{A}}^{P} q$ iff $p \mathop{\Longrightarrow}\limits_{\mathcal{A}}^{P'} q$ for all $p,q\in Q$, thus $p \mathop{\Longrightarrow}\limits_{\mathcal{A}}^{TP} q$ iff $p \mathop{\Longrightarrow}\limits_{\mathcal{A}}^{TP'} q$, for all $T\in SP^+(A)$, $p,q\in Q$.
  It follows that $\mathcal{R}(TP)=\mathcal{R}(TP')$, and thus $TP\sim_\mathcal{A}TP'$, ie. $\varphi_{\sim_\mathcal{A}}(TP)=\varphi_{\sim_\mathcal{A}}(T)\varphi_{\sim_\mathcal{A}}(P)=\varphi_{\sim_\mathcal{A}}(T)\varphi_{\sim_\mathcal{A}}(P')=\varphi_{\sim_\mathcal{A}}(TP')$ for all $T\in SP^+(A)$.
  As $\varphi_{\sim_\mathcal{A}}:SP^+(A)\to SP^+(A)/\mathord{\sim_{\mathcal{A}}}$ is surjective, it follows that $x\varphi_{\sim_\mathcal{A}}(P)=x\varphi_{\sim_\mathcal{A}}(P')$ (and, using symmetrical arguments, $\varphi_{\sim_\mathcal{A}}(P)x=\varphi_{\sim_\mathcal{A}}(P')x$) for all $x\in S$, thus $\varphi_{\sim_\mathcal{A}}(P)\sim_\text{seq}^{SP^+(A)}\varphi_{\sim_\mathcal{A}}(P')$.
\end{proof}

\begin{lem}
  \label{lem:equivSeq}
  $\sim_\text{seq}^{SP^+(A)/\mathord{\sim_{\mathcal{A}}}}$ has finite index, and each of its equivalence classes is a $\parallel$-rational set of $SP^+(A)/\mathord{\sim_{\mathcal{A}}}$.
\end{lem}

\begin{proof}
    As $\sim_\emptyset^{SP^+(A)}$ has finite index, then so is $\sim_\emptyset^{SP^+(A)/\mathord{\sim_{\mathcal{A}}}}$, and thus $\sim_\text{seq}^{SP^+(A)/\mathord\sim_\mathcal{A}}$ as a consequence of Lemma~\ref{lem:simEmptyImpliqueSimSeq}.
  Each equivalence class of $\sim_\text{seq}^{SP^+(A)/\mathord\sim_\mathcal{A}}$ is a finite union of equivalence classes of $\sim_\emptyset^{SP^+(A)/\mathord{\sim_{\mathcal{A}}}}$.
  As the class of $\parallel$-rational sets is closed under finite union, it follows from Lemma~\ref{lem:simEmptyParallelRational} that each equivalence class of $\sim_\text{seq}^{SP^+(A)/\mathord\sim_\mathcal{A}}$ is $\parallel$-rational.
\end{proof}

We have $SP^+(A)-L=\bigcup_{D\in\mathcal{P}(Q^2)\atop D\cap I\times F=\emptyset}\varphi^{-1}_{\sim_\mathcal{A}}(\Delta_D^{SP^+(A)/\mathord\sim_\mathcal{A}})$. Since the class of rational sets of $SP^+(A)$ is closed under finite union, it suffices to show that $\varphi^{-1}_{\sim_\mathcal{A}}(\Delta_D^{SP^+(A)/\mathord\sim_\mathcal{A}})$ is a rational set of $SP^+(A)$ for each $D\in\mathcal{P}(Q^2)$ in order to show that $SP^+(A)-L$ is a rational set of $SP^+(A)$. This will be achieved by Lemma~\ref{lem:revImageVarphiClasse} below.

We need to introduce the notion of \emph{$\parallel$-quotient} of a language. Let $L,L'\subseteq SP(A)$. The \emph{$\parallel$-quotient} of $L'$ by $L$ is $$\quotientparallel{L'}{L}=\{P\in SP(A) : \exists P'\in L\text{ such that } P\parallel P'\in L'\}$$

\begin{lem}
  \label{lem:quotientSemilinear}
  Let $A$ be an alphabet, and $X$, $Y$ be two $\parallel$-rational languages of $A^\circledast$. Then $\quotientparallel{Y}{X}$ is $\parallel$-rational.
\end{lem}

\begin{proof}
  By Theorem~\ref{th:ratSemLin}, $X$ and $Y$ are semi-linear. 
  By Theorem~\ref{th:PresburgerSemilinear}, $X$ and $Y$ are also Presburger sets of some formul{\ae}\ $\varphi_X(x_1,\dots,x_n)$ and $\varphi_Y(y_1,\dots,y_n)$. Up to a change in variables names we can assume that the free variables of both formul{\ae}\ are disjoint. Then $\quotientparallel{Y}{X}$ is the Presburger set of $$\varphi_{\quotientparallel{Y}{X}}(z_1,\dots,z_n)=\exists x_1,\dots,x_n,y_1,\dots,y_n\ \varphi_X(x_1,\dots,x_n)\land\varphi_Y(y_1,\dots,y_n)\land_{i\in[n]}z_i+x_i=y_i$$
  Using Theorems~\ref{th:PresburgerSemilinear} and~\ref{th:ratSemLin} again, $\quotientparallel{Y}{X}$ is $\parallel$-rational.
\end{proof}

\begin{lem}
  \label{lem:revImageRatEltFinite}
  Let $\varphi: SP^+(A)\to S$ be a morphism of sp-algebras.
  If $S$ is finite, then $\varphi^{-1}(s)$ is a regular set of $SP^+(A)$, for all $s\in S$.
\end{lem}

\begin{proof}
  We build an automaton $\mathcal{B}_s$ such that $L(\mathcal{B}_s)=\varphi^{-1}(s)$.
 The construction is a generalization of the well-know construction from finite semigroups to automata for finite words.
 Consider all the elements of $S$ as the states of $\mathcal{B}_s$, with one new state $1$. The unique initial state is $1$, the unique final state is $s$. Furthermore, in the following, $1$ is considered as a neutral element for both sequential and parallel product regarding the definition of transitions
 The sequential transitions are as usual: for each state $t$ and letter $a$, add a sequential transition from $s$ to $t\varphi(a)$.
 Let us deal now with the parallel product. For each $t\in S$, add two new states $\overline{t}_1$ and $\overline{t}_2$, a new letter $\underline{t}$ in the alphabet of $\mathcal{B}_s$, and a sequential transition $(\overline{t}_1,\underline{t},\overline{t}_2)$. We name the states of the form $\overline{t}_1$ and $\overline{t}_2$ \emph{special states}, and the new letters of the form $\underline{t}$ \emph{special letters}. For each $t\in S$, add a fork transition $(t,\{1,1,\overline{t}_1\})$. For each $t,u,v\in S$, add a join transition $(\{\overline{v}_2,t,u\},v\cdot(t\parallel u))$. Let $B$ be the set of special letters. Consider the projection $p:SP^+(A\cup B)\to SP^+(A)$ which removes special letters from posets. We have $p(L(\mathcal{B}_s))=\varphi^{-1}_{\sim_\mathcal{A}}(s)$. We now have to show that $p(L(\mathcal{B}_s))$ is rational (note that in general, projection does not preserve rationality). This is achieved by replacing, in the system of equations in the McNaughton-Yamada-like construction of a rational expression from an automaton (see~\cite[Section~4.2]{lodaya98kleene}), the equations of the form $\xi_{\overline{s}_1,\overline{s}_2}=\underline{s}$, whose solution is $L_{p,q}$, by $\xi_{\overline{s}_1,\overline{s}_2}=\{\emptyset\}$: those replacements do not affect the form of the system, whose solution remains rational.
\end{proof}

Lemma~\ref{lem:revImageRatEltFinite} proves that  $\varphi^{-1}_{\sim_\mathcal{A}}(s)$ is regular, for all $s\in SP^+(A)/\mathord\sim_\mathcal{A}$, when $SP^+(A)/\mathord\sim_\mathcal{A}$ is finite. We are now going to prove that $\varphi^{-1}_{\sim_\mathcal{A}}(\Delta_D^{SP^+(A)/\mathord\sim_\mathcal{A}})$ is regular for every equivalence class $\Delta_D^{SP^+(A)/\mathord\sim_\mathcal{A}}$ of  $\sim_\emptyset^{SP^+(A)/\mathord\sim_\mathcal{A}}$, even when $SP^+(A)/\mathord\sim_\mathcal{A}$ is not finite. The idea is to build an automaton $\mathcal{B}_{\Delta_D^{SP^+(A)/\mathord\sim_\mathcal{A}}}$ as in the proof of Lemma~\ref{lem:revImageRatEltFinite}, by showing that a finite subset of $SP^+(A)/\mathord\sim_\mathcal{A}$ can be used for the states of $\mathcal{B}_{\Delta_D^{SP^+(A)/\mathord\sim_\mathcal{A}}}$. Actually we prove this by translating the problem into an sp-algebra $\mathbb{N}^{k*}$ with more properties than $SP^+(A)/\mathord\sim_\mathcal{A}$. Very informally speaking, denote by $\{g_1,\dots,g_k\}$ the set of equivalence classes of $\sim_\mathcal{A}$ containing a sequential poset of $SP^+(A)$ (which is finite by Lemma~\ref{lem:finiteIndex}). For every $P\in SP^+(A)$ whose maximal parallel factorization is $P=P_1\parallel\dots\parallel P_n$, the morphism $\mu:SP^+(A)\to\mathbb{N}^{k*}$ enables the count of $\#i$, $i\in[n]$, such that $P_i\in g_j$, for every $j\in[k]$. Also, every language recognized by $SP^+(A)/\mathord\sim_\mathcal{A}$ is recognized by $\mathbb{N}^{k*}$.

\medskip

Let $G=\{g_1,\dots,g_k\}=\{ \varphi_{\sim_\mathcal{A}}(P) : P\in Seq(SP^+(A))\}$. Then $G$ is a (finite, by Lemma~\ref{lem:finiteIndex}) generator of $(SP^+(A)/\mathord\sim_\mathcal{A},\parallel)$. We may suppose, by Proposition~\ref{prop:seqSepAuto}, that $\mathcal{A}$ is sequentially separated. Thus that 
the elements of $G$ are indecomposable with respect to the parallel product, that is to say, each $g_i\in G$ can not be written $g_i=s\parallel s'$ with $s,s'\in SP^+(A)/\mathord\sim_\mathcal{A}$.

Denote by $(\mathbb{N}^{k*},+)$ the commutative semigroup whose elements are $k$-tuples of non-negative integers, without $(0,\dots,0)$. It is generated by the $k$-tuples with all components set to 0, except one which is set to 1. For short we denote by $1^i$ the element of the generator of $\mathbb{N}^{k*}$ with the $i^\text{th}$ component set to~1. The parallel product $+$ of $(\mathbb{N}^{k*},+)$ is the sum componentwise. We define a morphism of commutative semigroups $\psi:(\mathbb{N}^{k*},+)\to (SP^+(A)/\mathord\sim_\mathcal{A},\parallel)$ by $\psi(1^i)=g_i$ for all $i\in[k]$. Note that $\psi$ is surjective, and that $\psi^{-1}(g_i)=\{1^i\}$ for all $i\in[k]$. As a consequence $\psi^{-1}(ss')$ is a singleton for all $ss'\in SP^+(A)/\mathord\sim_\mathcal{A}$. Now we equip $(\mathbb{N}^{k*},+)$ with a sequential product, by setting, for all $n_1,n_2\in\mathbb{N}^{k*}$, $n_1n_2=\psi^{-1}(\psi(n_1)\psi(n_2))$. This sequential product equips $\mathbb{N}^{k*}$ with a structure of semigroup since $(n_1n_2)n_3=\psi^{-1}(\psi(n_1n_2)s_3)=\psi^{-1}(\psi(\psi^{-1}(s_1s_2))s_3)=\psi^{-1}(s_1s_2s_3)=n_1(n_2n_3)$, considering $\psi(n_i)=s_i$ for all $i\in[3]$. Thus, $\mathbb{N}^{k*}$ equipped with its parallel and sequential products is an sp-algebra. Observe that $\psi(n_1n_2)=\psi(n_1)\psi(n_2)$. Now, we define a morphism of sp-algebras $\mu:SP^+(A)\to \mathbb{N}^{k*}$ by $\mu(a)=\psi^{-1}\varphi_{\sim_\mathcal{A}}(a)$ for all $a\in A$. The diagram of Figure~\ref{fig:diagram} sums up the situation.
\begin{figure}
  \begin{diagram}
    SP^+(A)         & \rTo^{\varphi_{\sim_\mathcal{A}}}   & SP^+(A)/\mathord\sim_\mathcal{A} \\
    \dTo^\mu        & \ruDashto^\psi                      &                          \\
    \mathbb{N}^{k*} &                                     &
  \end{diagram}
  \caption{The morphisms between the sp-algebras. Full arrows represent morphisms of sp-algebras, and dashed arrows morphisms of commutative semigroups.}
  \label{fig:diagram}
\end{figure}

\begin{lem}
  \label{lem:phiMuPsi}
  For all $s\in SP^+(A)/\mathord\sim_\mathcal{A}$, $\varphi^{-1}_{\sim_\mathcal{A}}(s)=\mu^{-1}\psi^{-1}(s)
%\label{eq:phiMuPsi}
$.
\end{lem}

\begin{proof}
  We show by induction on $P\in SP^+(A)$ that if $P$ belongs to one side of the equality then it also belongs to the other.
Let $P=a\in A$. If $P\in\varphi^{-1}_{\sim_\mathcal{A}}(s)$ (thus $s\in G$), by definition $\mu(P)=\psi^{-1}(s)$. If $\mu(P)\in\psi^{-1}(s)$ then $s=\varphi_{\sim_\mathcal{A}}(P)$.
Assume now that $P$ has the form $P=P_1\parallel P_2$ for some $P_1,P_2\in SP^+(A)$. If $P\in\varphi^{-1}_{\sim_\mathcal{A}}(s)$ then $\varphi_{\sim_\mathcal{A}}(P)=s=\varphi_{\sim_\mathcal{A}}(P_1)\parallel\varphi_{\sim_\mathcal{A}}(P_2)$. Denote by $s_i=\varphi_{\sim_\mathcal{A}}(P_i)$, $i\in[2]$. By induction hypothesis $P_i\in\mu^{-1}\psi^{-1}(s_i)$. Thus $P_1\parallel P_2\in(\mu^{-1}\psi^{-1}(s_1))\parallel(\mu^{-1}\psi^{-1}(s_2))\subseteq \mu^{-1}\psi^{-1}(s)$. On the other side, if $P\in \mu^{-1}\psi^{-1}(s)$, it follows from the induction hypothesis that $\psi\mu(P_i)=s_i$, $i\in[2]$. As $\psi\mu(P)=\psi\mu(P_1\parallel P_2)=\psi\mu(P_1)\parallel\psi\mu(P_2)=s_1\parallel s_2$, we have $P\in\mu^{-1}\psi^{-1}(s_1\parallel s_2)$, thus $s_1\parallel s_2=s$ and thus $P\in\varphi^{-1}_{\sim_\mathcal{A}}(s)$.
Finally assume that $P$ has the form $P=P_1P_2$ for some $P_1,P_2\in SP^+(A)$. If $P\in\varphi^{-1}_{\sim_\mathcal{A}}(s)$ we proceed as in the case $P=P_1\parallel P_2$, observing that $\psi^{-1}(s_1)\psi^{-1}(s_2)\subseteq\psi^{-1}(s)$ if $s_1s_2=s$: if $n_i\in\psi^{-1}(s_i)$, $i\in[2]$, then by definition $n_1n_2=\psi^{-1}(s_1s_2)=\psi^{-1}(s)$. On the other side, if $P\in \mu^{-1}\psi^{-1}(s)$, it follows from the induction hypothesis that $\psi\mu(P_i)=s_i$, $i\in[2]$. We have $\psi\mu(P)=\psi\mu(P_1P_2)=\psi(\mu(P_1)\mu(P_2))=\psi(\mu(P_1))\psi(\mu(P_2))=s_1s_2$,
so $s_1s_2=s$. As $\varphi_{\sim_\mathcal{A}}(P)=\varphi_{\sim_\mathcal{A}}(P_1)\varphi_{\sim_\mathcal{A}}(P_2)=s_1s_2=s$. 
\end{proof}

\begin{lem}
  \label{lem:muSurj}
   $\mu$ is a surjective morphism of sp-algebras.
\end{lem}

\begin{proof}
  Let $n\in\mathbb{N}^{k*}$. If $n=1^i$ for some $i\in[k]$ then $n=\psi^{-1}(g_i)$. As $\varphi_{\sim_\mathcal{A}}$ is surjective there exists $P\in\varphi^{-1}_{\sim_\mathcal{A}}(g_i)$, thus $P\in\mu^{-1}(n)$. Otherwise, $n=1^{i_1}+\dots+1^{i_r}$ for some $r>1$ and $i_1,\dots i_r\in[k]$. As for all $j\in[r]$, $1^{i_j}=\psi^{-1}(g_{i_j})$, there is some $P_j\in\varphi^{-1}_{\sim_\mathcal{A}}(g_{i_j})$, and $\mu(P_1\parallel\dots\parallel P_r)=\mu(P_1)+\dots+\mu(P_r)=1^{i_1}+\dots+1^{i_r}=n$. 
\end{proof}

\begin{lem}
  \label{lem:muOK}
  $\mu$ verifies the conditions of Lemma~\ref{lem:simEmptyParallelRational}.
\end{lem}

\begin{proof}
  First, $\mu$ is surjective by Lemma~\ref{lem:muSurj}.
  That $\mu$ recognizes $L_{p,q}$ for all $(p,q)\in Q^2$ is a consequence of Lemma~\ref{lem:phiMuPsi}.
  Because $\mu(L_{p,q})=\mathop\bigcup_{X\in \mathcal{F}_{p,q}}\mathop\parallel_{(r,s)\in X} \mu(Seq(L_{r,s}))$, and as $\mu(Seq(SP^+(A)))$ is finite, then $\mu(L_{p,q})$ is a $\parallel$-rational set of $\mathbb{N}^{k*}$.
\end{proof}

It follows from Lemmas~\ref{lem:phiMuPsi} and~\ref{lem:equivEmptyRecLpq} that the equivalence relation $\sim_\emptyset^{\mathbb{N}^{k*}}$ can be defined over the elements of $\mathbb{N}^{k*}$.
Furthermore, recall that $\mu$ is surjective by Lemma~\ref{lem:muSurj}.
As a consequence of Lemmas~\ref{lem:phiMuPsi}, \ref{lem:muOK} and~\ref{lem:simEmptyParallelRational}, each of the equivalence class $\Delta_D^{\mathbb{N}^{k*}}$ of $\sim_\emptyset^{\mathbb{N}^{k*}}$ is a $\parallel$-rational of $\mathbb{N}^{k*}$, and thus, by Theorem~\ref{th:ratSemLin}, has the form $\Delta_D^{\mathbb{N}^{k*}}=\cup_{i\in I_D}(a_{D,i}+ B_{D,i}^\circledast)$ for some finite set $I_D$, $a_{D,i}\in\mathbb{N}^{k*}$, $B_{D,i}$ some finite part of $\mathbb{N}^{k*}$. For all $i\in I_D$ set $\Delta_{D,i}^{\mathbb{N}^{k*}}=a_{D,i}+ B_{D,i}^\circledast$. We may assume that all the $\Delta_{D,i}^{\mathbb{N}^{k*}}$ are pairwise disjoint~\cite[Theorem~IV]{EilSch69}. Note that the decomposition of $\Delta_D^{\mathbb{N}^{k*}}$ into $\Delta_D^{\mathbb{N}^{k*}}=\cup_{i\in I_D}\Delta_{D,i}^{\mathbb{N}^{k*}}$ is not unique. This decomposition may influence the constructions below, but not the main result (Lemma~\ref{lem:revImageVarphiClasse}).

The following lemma links the equivalence classes of $\sim_\emptyset^{SP^+(A)/\mathord\sim_\mathcal{A}}$ and $\sim_\emptyset^{\mathbb{N}^{k*}}$.
\begin{lem}
  For all $n_1,n_2\in\mathbb{N}^{k*}$, $n_1\sim_\emptyset^{\mathbb{N}^{k*}} n_2$ iff $\psi(n_1)\sim_\emptyset^{SP^+(A)/\mathord\sim_\mathcal{A}} \psi(n_2)$.
\end{lem}

\begin{proof}
  Assume first there exist $P_1\in\mu^{-1}(n_1)$, $P_2\in\mu^{-1}(n_2)$ with $P_1\sim_\emptyset^{SP^+(A)} P_2$. As a consequence of Lemma~\ref{lem:phiMuPsi} we have $\varphi_{\sim_\mathcal{A}}(P_i)=\psi\mu(P_i)$ for all $i\in[2]$, thus $\psi\mu(P_1)\sim_\emptyset^{SP^+(A)/\mathord\sim_\mathcal{A}} \psi\mu(P_2)$.
Assume now, for the inclusion from right to left, that there exist $P_i\in\varphi_{\sim_\mathcal{A}}^{-1}(\psi(n_i))$ for all $i\in[2]$, with $P_1\sim_\emptyset^{SP^+(A)} P_2$, but $P'_1\not\sim_\emptyset^{SP^+(A)} P'_2$ for all $P'_1\in\mu^{-1}(n_1)$, $P'_2\in\mu^{-1}(n_2)$. As a consequence of Lemma~\ref{lem:phiMuPsi} we have, for all $i\in[2]$, $\psi\mu(P'_i)=\psi(n_i)$, thus $P_i,P'_i\in\varphi^{-1}_{\sim_\mathcal{A}}(\psi(n_i))$ and thus $P_i\sim_\emptyset^{SP^+(A)} P'_i$. As $P'_1\sim_\emptyset^{SP^+(A)} P_1\sim_\emptyset^{SP^+(A)} P_2\sim_\emptyset^{SP^+(A)} P'_2$ we have $P'_1\sim_\emptyset^{SP^+(A)} P'_2$, which is a contradiction.
\end{proof}

\begin{lem}
  \label{lem:pdtSeqEq}
  For all $t,t'\in\mathbb{N}^{k*}$, if $t\sim_\emptyset^{\mathbb{N}^{k*}} t'$ then $t\sim_\text{seq}^{\mathbb{N}^{k*}}t'$.
\end{lem}

\begin{proof}
  We first show $\psi(tn)=\psi(t'n)$ and $\psi(nt)=\psi(nt')$. Indeed, let $P,P',N\in SP^+(A)$ such that $\mu(P)=t$, $\mu(P')=t'$ and $\mu(N)=n$. If $t\sim_\emptyset^{\mathbb{N}^{k*}}t'$ then for all $(p,q)\in Q^2$, $P$ is the label of a path from $p$ to $q$ in $\mathcal{A}$ iff so is $P'$. As a consequence, for all $(p,q)\in Q^2$, $PN$ is the label of a path from $p$ to $q$ in $\mathcal{A}$ iff so is $P'N$. Thus $\mathcal{R}(PN)=\mathcal{R}(P'N)$. This implies $\varphi_{\sim_\mathcal{A}}(PN)=\varphi_{\sim_\mathcal{A}}(P'N)$, thus $\psi(tn)=\psi(t'n)$. Now, by definition $tn=\psi^{-1}(\psi(t)\psi(n))$, thus $\psi(tn)=\psi(t)\psi(n)=g_i=\psi(t'n)$ for some $g_i\in G$, thus $tn=t'n=1^i$. We show that $nt=nt'$ using symmetrical arguments.
\end{proof}

\begin{lem}
  \label{lem:revImageMuRatElt}
  For all $n\in\mathbb{N}^{k*}$, $\mu^{-1}(n)$ is a regular set of $SP^+(A)$.
\end{lem}

The construction given in the proof of Lemma~\ref{lem:revImageMuRatElt} below is illustrated by Example~\ref{ex:complement} located at the end of this section.

\begin{proof}
  The lemma is achieved by constructing an automaton $\mathcal{B}_n$ from $\mathbb{N}^{k*}$. Take $m$ an element of $\mathbb{N}^{k*}$ which is greater than $n$ and all $a_{D,i}+b_i$, for all $b_i\in B_{D,i}$, $i\in I_D$, $D\in\mathcal{P}(Q^2)$. The finite set $S$ of states of $\mathcal{B}_n$ consists in
\begin{itemize}
\item $S_1=\{ x\in\mathbb{N}^{k} : x\leq m\}$; 
\item for each $\Delta_{D,i}^{\mathbb{N}^{k*}}=a_{D,i}+B_{D,i}^\circledast$, a new state $\Delta_{D,i}$ (set $S_2=\{ \Delta_{D,i} : i\in I_D, D\in\mathcal{P}(Q^2)\}$);
\item for each element $x\in S_1\cup S_2$, two additional special states $\overline{x_1}$ and $\overline{x_2}$ (set $S_3=\{\overline{x_i} : i\in[2], x\in S_1\cup S_2\}$) and a new letter $\underline{x}$.
\end{itemize}
For uniformity with the construction given in the proof of Lemma~\ref{lem:revImageRatEltFinite}, we denote by $1=(0,\dots,0)$. We have $S=S_1\cup S_2\cup S_3$. 
For all $s,s'\in S_1$, $\Delta_{D,i}\in S_2$, $s'',s'''\in (S_1\cup S_2)-\{1\}$, define $\circ:(S_1\cup S_2)^2\to S_1\cup S_2$ 
such that $1$ is a neutral element for $\circ$,
 $\Delta_{D,i}\circ s'''=a_{D,i}\circ s'''$, $s''\circ \Delta_{D,i}=s''\circ a_{D,i}$, $s\circ s'=s\cdot s'$.
Note that $x\circ y\in S_2$ iff one of $x,y$ belongs to $S_2$ and the other is~$1$.
It is just verification to check that $\circ$ is associative, as a consequence of the associativity of~$\cdot$.
The sequential and fork transitions whose source belongs to $S_1\cup S_2$ are defined as in the construction in the proof of Lemma~\ref{lem:revImageRatEltFinite}, by replacing the sequential product of $SP^+(A)/\mathord\sim_\mathcal{A}$ by $\circ$.
For all $s,s'\in S_1\cup S_2$, add a sequential transition $(\overline{s}_1,\underline{s},\overline{s}_2)$. Define also
$$s\oplus s'=
\begin{cases}
  \delta(s+s')&\text{if }s,s'\in S_1,\\
  \Delta_{D,i}&\text{if }s=\Delta_{D,i}\text{ and } s'\in B_{D,i},\\
  \Delta_{D,i}&\text{if }s\in B_{D,i}\text{ and } s'=\Delta_{D,i},\\
  \text{undefined }&\text{otherwise}
\end{cases}
$$
where $\delta:\mathbb{N}^k\to S_1\cup S_2$ is given by, for all $n\in\mathbb{N}^k$,
$$\delta(n)=
\begin{cases}
  n & \text{if }n\leq m,\\
  \Delta_{D,i} & \text{if }n\in\Delta_{D,i}^{\mathbb{N}^{k*}}\text{ and not }n\leq m.
\end{cases}$$
 The join transitions are defined as follows. For each $s,t\in (S_1\cup S_2)-\{1\}$, $\overline{u}_2\in S_3$, add a join transition $(\{\overline{u}_2,s,t\},u\circ(s\oplus t))$ if $u\circ(s\oplus t)$ is defined. The unique initial state of $\mathcal{B}_n$ is $1$, and its unique final state is $n$.
 From now we slightly change our notation for simplicity: we denote by $v\mathop{\Longrightarrow}\limits_{\mathcal{B}_n}^{P} x$ the existence of a poset $P'\in SP^+(A\cup B)$ such that $p(P')=P$ ($p$ is defined as in the proof of Lemma~\ref{lem:revImageRatEltFinite}) and of path in $\mathcal{B}_n$ from $v$ to $x$ labeled by $P'$. 
We claim that, for all $P\in SP^+(A)$, $v,x\in S_1\cup S_2$,  $v\mathop{\Longrightarrow}\limits_{\mathcal{B}_n}^{P} x$ iff $x=v\circ\delta\mu (P)$. First, the implication from left to right. We proceed by induction on $P$. If $P=a\in A$ then necessarily $\mu(P)\leq m$. By construction there is a sequential transition labeled by $a$ from $v$ to $x$ iff $x=v\circ\mu(a)=v\circ\delta\mu(a)$. Assume now $P=P_1P_2$ for some $P_1,P_2\in SP^+(A)$. A path $\gamma:v\mathop{\Longrightarrow}\limits_{\mathcal{B}_n}^{P} x$ can be decomposed into $\gamma:v\mathop{\Longrightarrow}\limits_{\mathcal{B}_n}^{P_1} y\mathop{\Longrightarrow}\limits_{\mathcal{B}_n}^{P_2} x$, and by induction hypothesis we have $y=v\circ\delta\mu(P_1)$, thus $x=(v\circ\delta\mu(P_1))\circ\delta\mu(P_2)=v\circ(\delta\mu(P_1)\circ\delta\mu(P_2))$ with the help of the associativity of $\circ$. As a consequence of Lemma~\ref{lem:pdtSeqEq} and by definition of $\delta$, we have $(\delta (n))\circ x=n\cdot x$ for all $n,x\in\mathbb{N}^{k*}$, thus $\delta\mu(P_1)\circ\delta\mu(P_2)=\mu(P_1)\cdot\mu(P_2)=\mu(P_1P_2)=\delta\mu(P_1P_2)$ since $\mu(P_1P_2)\leq m$ because $\mu(P_1P_2)$ has the form $1^i$ for some $i\in[k]$. Finally assume $P=P_1\parallel P_2$ for some $P_1,P_2\in SP^+(A)$. If there is a path $\gamma:v\mathop{\Longrightarrow}\limits_{\mathcal{B}_n}^{P} x$, then by construction and with the help of the induction hypothesis it has the form $\gamma=(v,\{1,1,\overline{v}_1\})(\gamma_1\parallel\gamma_2\parallel\gamma_3)(\{\delta\mu(P'_1),\delta\mu(P'_2),\overline{v}_2\},x)$ where $\gamma_i:1\mathop{\Longrightarrow}\limits_{\mathcal{B}_n}^{P'_i}\delta\mu(P'_i)$ for all $i\in[2]$, and for some $P'_1,P'_2\in SP^+(A)$ such that $P=P'_1\parallel P'_2$, and $\gamma_3$ is the path consisting of the sequential transition $(\overline{v}_1,\underline{v},\overline{v}_2)$. By definition of the join transitions we have $x=v\circ(\delta\mu(P'_1)\oplus\delta\mu(P'_2))$. Thus $\delta\mu(P'_1)\oplus\delta\mu(P'_2)$ is defined and we have three cases. In the first case $\delta\mu(P'_1),\delta\mu(P'_2)\in S_1$ and we have $\delta\mu(P'_1)\oplus\delta\mu(P'_2)=\delta(\mu(P'_1)+\mu(P'_2))=\delta\mu(P)$. Up to a symmetry, the second and third cases are similar, so assume wlog. we are in the second case: $\delta\mu(P'_1)\oplus\delta\mu(P'_2)=\Delta_{D,i}$ with $\delta\mu(P'_1)=\Delta_{D,i}$ and $\delta\mu(P'_2)\in B_{D,i}$. Necessarily $\mu(P'_1)=a_{D,i}+b_{i,1}+\dots+b_{i,r}$ for some $r\in\mathbb{N}$, $b_{i,j}\in B_{D,i}$ for all $j\in[r]$, and $\mu(P'_2)=b_{i,r+1}$ for some $b_{i,r+1}\in B_{D,i}$. Thus $\mu(P)=a_{D,i}+b_{i,1}+\dots+b_{i,r}+b_{i,r+1}$ and $\delta\mu(P)=\Delta_{D,i}$. Let us turn now to the implication from right to left. The cases $P=a$ and $P=P_1P_2$ for some $P_1,P_2\in SP^+(A)$ are as above, so assume $P=P_1\parallel P_2$. Up to a parallel refactorization of $P$ we may assume, if $\delta\mu(P)=\Delta_{D,i}$ for some $\Delta_{D,i}$, that $\mu(P_1)=a_{D,i}+b_{i,1}+\dots+b_{i,r}$ for some $r\in\mathbb{N}$, $b_{i,j}\in B_{D,i}$ for all $j\in[r]$, and $\mu(P_2)=b_{i,r+1}$ for some $b_{i,r+1}\in B_{D,i}$. So assume first $\delta\mu(P)=\Delta_{D,i}$: either $\mu(P_1),\mu(P_2)\in S_1$, or $\delta\mu(P_1)=\Delta_{D,i}$ and $\mu(P_2)\in S_1$. In the first case, for all $i\in[2]$, $\mu(P_i)=\delta\mu(P_i)$, and by induction hypothesis there is a path $\gamma_i:1\mathop{\Longrightarrow}\limits_{\mathcal{B}_n}^{P_i}1\circ\delta\mu(P_i)=\delta\mu(P_i)=\mu(P_i)$. By construction there is a fork transition $f=(v,\{1,1,\overline{v}_1\})$, a sequential transition $t=(\overline{v}_1,\underline{v},\overline{v}_2)$ and a join transition $j=(\{\mu(P_1),\mu(P_2),\overline{v}_2\},v\circ(\mu(P_1)\oplus\mu(P_2))$ with $v\circ(\mu(P_1)\oplus\mu(P_2))=v\circ\delta(\mu(P_1)+\mu(P_2))=v\circ\delta\mu(P)$, thus $f(\gamma_1\parallel\gamma_2\parallel t)j$ forms a path $\gamma:v\mathop{\Longrightarrow}\limits_{\mathcal{B}_n}^{P} v\circ\delta\mu(P)$. In the second case, by induction hypothesis there exist a path  $\gamma_1:1\mathop{\Longrightarrow}\limits_{\mathcal{B}_n}^{P_i}1\circ\delta\mu(P_1)=\Delta_{D,i}$ and a path $\gamma_2:1\mathop{\Longrightarrow}\limits_{\mathcal{B}_n}^{P_i}1\circ\delta\mu(P_2)=\mu(P_2)$.  By construction there is a fork transition $f=(v,\{1,1,\overline{v}_1\})$, a sequential transition $t=(\overline{v}_1,\underline{v},\overline{v}_2)$ and a join transition $j=(\{\Delta_{D,i},\mu(P_2),\overline{v}_2\},v\circ(\Delta_{D,i}\oplus\mu(P_2)))$ that can be used to form a path  $\gamma:v\mathop{\Longrightarrow}\limits_{\mathcal{B}_n}^{P} v\circ\delta\mu(P)$ because $\Delta_{D,i}\oplus\mu(P_2)=\Delta_{D,i}=\mu(P)$. Finally the case $\delta\mu(P)\in S_1$ is identical to the case $\delta\mu(P)=\Delta_{D,i}$ with $\mu(P_1),\mu(P_2)\in S_1$.
\end{proof}

\begin{lem}
  \label{lem:revImageMuRatDelta}
  For all $D\in\mathcal{P}(Q^2)$, $i\in I_D$, $\mu^{-1}(\Delta_{D,i}^{\mathbb{N}^{k*}})$ is a regular set of $SP^+(A)$.
\end{lem}

\begin{proof}
  The construction is almost the same as in the proof of Lemma~\ref{lem:revImageMuRatElt}.
  We only change $m$ to be greater than all $a_{D,i}+b_i$, for all $b_i\in B_{D,i}$, $i\in I_D$, $D\in\mathcal{P}(Q^2)$, without considering $n$, and the final states are $\Delta_{D,i}$, and all states belonging to $a_{D,i}+B_{D,i}^\circledast$.
\end{proof}

\begin{lem}
  \label{lem:revImageVarphiClasse}
  For all $D\in\mathcal{P}(Q^2)$, $\varphi^{-1}_{\sim_\mathcal{A}}(\Delta_D^{SP^+(A)/\mathord\sim_\mathcal{A}})$ is a regular set of $SP^+(A)$. Similarly, for each equivalence class $c$ of $\sim_\text{seq}^{SP^+(A)/\mathord\sim_\mathcal{A}}$, $\varphi^{-1}_{\sim_\mathcal{A}}(c)$ is a regular set of $SP^+(A)$.
\end{lem}

\begin{proof}
  By Lemmas~\ref{lem:muOK} and~\ref{lem:equivEmptyRecLpq}, $\varphi^{-1}_{\sim_\mathcal{A}}(\Delta_D^{SP^+(A)/\mathord\sim_\mathcal{A}})=\Delta_D^{SP^+(A)}=\mu^{-1}(\Delta_D^{\mathbb{N}^{k*}})$.
  Because $\mu^{-1}(\Delta_D^{\mathbb{N}^{k*}})=\cup_{i\in I_D}\mu^{-1}(\Delta_{D,i}^{\mathbb{N}^{k*}})$, with $I_D$ finite, and regular sets are closed under finite union, it follows from Lemma~\ref{lem:revImageMuRatDelta} that  $\varphi^{-1}_{\sim_\mathcal{A}}(\Delta_D^{SP^+(A)/\mathord\sim_\mathcal{A}})$ is a regular set of $SP^+(A)$. As by Lemma~\ref{lem:simEmptyImpliqueSimSeq} an equivalence class $c$ of $\sim_\text{seq}^{SP^+(A)/\mathord\sim_\mathcal{A}}$ is a finite union of equivalence classes of $\sim_\emptyset^{SP^+(A)/\mathord\sim_\mathcal{A}}$,  $\varphi^{-1}_{\sim_\mathcal{A}}(c)$ is also a regular set of $SP^+(A)$.
\end{proof}

We now give an example illustrating the construction given in the proof of Lemma~\ref{lem:revImageMuRatElt}.

\begin{exa}
  \label{ex:complement}
  We consider the rational language $L=((aa)\parallel a)^\oplus a$ of $SP^+(A)$ with $A=\{a\}$, and the automaton $\mathcal{A}$ pictured in Figure~\ref{fig:exComplement} which verifies $L(\mathcal{A})=L$.
  \begin{figure}[htbp]
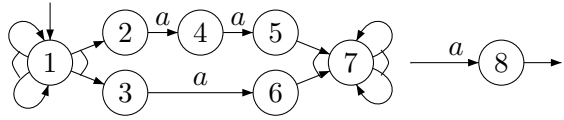

    \begin{center}
      \gasset{Nw=6,Nh=6,loopdiam=4}
      \begin{gpicture}
        \node(1)(-10,-12){$1$}\imark[iangle=90](1)
        \node(2)(0,-8){$2$}
        \node(3)(0,-16){$3$}
        \node(4)(20,-8){$5$}
        \node(5)(20,-16){$6$}
        \node(6)(10,-8){$4$}
        \node(7)(30,-12){$7$}
        \node[Nw=0,Nh=0,loopdiam=0](9)(38,-12){}
        \node(8)(50,-12){$8$}\fmark[fangle=0](8)
        \drawedge(2,6){$a$}
        \drawedge(6,4){$a$}
        \drawedge(3,5){$a$}
        \drawedge(1,2){}
        \drawedge(1,3){}
        \drawedge(9,8){$a$}
        \drawloop[loopangle=135](1){}
        {\gasset{AHnb=0,ATnb=1}
          \drawloop[loopangle=235](1){}}
        \drawcurve[AHnb=0](-6,-10.3)(-5,-12)(-6,-13.7)
        \drawcurve[AHnb=0](-14,-10.5)(-15,-12)(-14,-14)
        \drawedge(4,7){}
        \drawedge(5,7){}
        {\gasset{AHnb=0,ATnb=1}\drawloop[loopangle=45](7){}}
        \drawloop[loopangle=-45](7){}
        \drawcurve[AHnb=0](34,-10.3)(35,-12)(34,-13.7)
        \drawcurve[AHnb=0](26,-10.5)(25,-12)(26,-14)
      \end{gpicture}
      \caption{An automaton $\mathcal{A}$ with $L(\mathcal{A})=((aa)\parallel a)^\oplus a$.}
      \label{fig:exComplement}
    \end{center}
  \end{figure}
  We have
  $$
    \mathcal{F}_{(p,q)}=
    \begin{cases}
      \bigcup_{n,m\in\mathbb{N}\atop n+m>0}\{\{(1,7)^n,(2,5)^m,(3,6)^m\}\} & \text{if }(p,q)=(1,7);\\
      \{\{(p,q)\}\} & \text{if }(p,q)\in\{(2,4),(4,5),(2,5),\\&\hfill (3,6),(7,8),(1,8)\};\\
      \emptyset & \text{otherwise.}
    \end{cases}
  $$
  As stated in Lemma~\ref{lem:FqpParallelRat}, $\mathcal{F}_{(p,q)}$ is a $\parallel$-rational set of $(Q\times Q)^\oplus$, with $Q$ the set of states of $\mathcal{A}$, for all $(p,q)\in Q\times Q$. For example, $\mathcal{F}_{(1,7)}=((1,7)+(2,5)\parallel(3,6))^\oplus$.
  We now compute $SP^+(A)/\mathord{\sim_\mathcal{A}}$.
  First observe that ($k>0$ in the equalities below):
  \begin{align*}
    \bigcup_{x\in \mathcal{R}((aa)^{\parallel k})}\quotientparallel{\mathcal{F}_{(1,7)}}{x}=&\bigcup_{n,m\in\mathbb{N}\atop m\geq k}\{\{(1,7)^n,(2,5)^{m-k},(3,6)^m\}\}\\
    \bigcup_{x\in \mathcal{R}(a^{\parallel k})}\quotientparallel{\mathcal{F}_{(1,7)}}{x}=&\bigcup_{n,m\in\mathbb{N}\atop m\geq k}\{\{(1,7)^n,(2,5)^{m},(3,6)^{m-k}\}\}\\
    \bigcup_{x\in \mathcal{R}((a\parallel (aa))^{\parallel k})}\quotientparallel{\mathcal{F}_{(1,7)}}{x}=&\bigcup_{n,m\in\mathbb{N}}\{\{(1,7)^{n},(2,5)^{m},(3,6)^{m}\}\}\\
  \end{align*}
  %where $X(P)=\{\{ (p_1,q_1),\dots,(p_n,q_n) : p_i \mathop{\Longrightarrow}\limits_{\mathcal{A}}^{P_i} q_i \text{ for all }i\in[n]\} : P=P_1\parallel\dots\parallel P_n, P_i\in SP^+(A)\text{ for all }i\in[n]\}$, 
  with
  \begin{align*}
    \mathcal{R}(a^{\parallel k})=&\{x_1,\dots,x_k\}\text{ with } x_i\in\{(2,4),(3,6),(4,5),(7,8)\}\text{ for all }i\in[k]\\
    \mathcal{R}((aa)^{\parallel k})=&\{\{(2,5)^k\}\}\\
    \mathcal{R}((a\parallel (aa))^{\parallel k})=&\{\{(1,7)^k\}\}
  \end{align*}
  As $\emptyset\in \bigcup_{x\in \mathcal{R}((a\parallel (aa))^{\parallel k})}\quotientparallel{\mathcal{F}_{(1,7)}}{x}$ (take $n=m=0$) then $1 \mathop{\Longrightarrow}\limits_{\mathcal{A}}^{(a\parallel (aa))^{\parallel k}} 7$ for all $k>0$.
  On the other side, as $\emptyset\not\in\bigcup_{x\in \mathcal{R}(a^{\parallel k})}\quotientparallel{\mathcal{F}_{(1,7)}}{x}$ (resp. $\emptyset\not\in\bigcup_{x\in \mathcal{R}((aa)^{\parallel k})}\quotientparallel{\mathcal{F}_{(1,7)}}{x}$), then for all $k>0$, not $1 \mathop{\Longrightarrow}\limits_{\mathcal{A}}^{a^{\parallel k}} 7$ for all $k>0$ (resp. not $1 \mathop{\Longrightarrow}\limits_{\mathcal{A}}^{(aa)^{\parallel k}} 7$).
  As $\bigcup_{x\in \mathcal{R}((aa)^{\parallel k})}\quotientparallel{\mathcal{F}_{(1,7)}}{x}\not=\bigcup_{x\in \mathcal{R}((aa)^{\parallel k'})}\quotientparallel{\mathcal{F}_{(1,7)}}{x}$ and $\bigcup_{x\in \mathcal{R}(a^{\parallel k})}\quotientparallel{\mathcal{F}_{(1,7)}}{x}\not=\bigcup_{x\in \mathcal{R}(a^{\parallel k'})}\quotientparallel{\mathcal{F}_{(1,7)}}{x}$ for all $k,k'>0$ with $k\not=k'$, then $SP^+(A)/\mathord{\sim_\mathcal{A}}$ has not finite index. 
  Actually $SP^+(A)/\mathord{\sim_\mathcal{A}}$ is composed of the following equivalence classes (recall that $\varphi_{\sim_{\mathcal{A}}}:SP^+(A)\to SP^+(A)/\mathord\sim_\mathcal{A}$ is the morphism which associates to each poset $P\in SP^+(A)$ its equivalence class in $SP^+(A)/\mathord\sim_\mathcal{A}$):
  \begin{itemize}
  \item for all $k>0$, one class denoted by $a^{\parallel k}$, such that $\varphi_{\sim_{\mathcal{A}}}(a^{\parallel k})=a^{\parallel k}$;
  \item for all $k>0$, one class denoted by $(aa)^{\parallel k}$, such that $\varphi_{\sim_{\mathcal{A}}}((aa)^{\parallel k})=(aa)^{\parallel k}$;
  \item one class denoted by $(aa)\parallel a$, such that $\varphi_{\sim_{\mathcal{A}}}(((aa)\parallel a)^{\parallel k})=(aa)\parallel a$ for all $k>0$;
  \item one class denoted by $((aa)\parallel a)a$, such that $\varphi_{\sim_{\mathcal{A}}}((((aa)\parallel a)^{\parallel k})a)=((aa)\parallel a)a$ for all $k>0$;
  \item one class denoted by $0$, such that $\varphi_{\sim_{\mathcal{A}}}(P)=0$ for all $P\in SP^+(A)$ which are not mentioned above.
  \end{itemize}
  The sp-algebra $SP^+(A)/\mathord\sim_\mathcal{A}$ is equipped with the parallel product $\parallel$ verifying
  \begin{align*}
    a^{\parallel k}\parallel (aa)^{\parallel k'} =&
    \begin{cases}
      a^{\parallel (k-k')} & \text{if }k>k'\\
      (aa)^{\parallel (k'-k)} & \text{if }k'>k\\
      (aa)\parallel a & \text{otherwise}
    \end{cases}
    \end{align*}
    \begin{multicols}{2}
      \begin{itemize}
    \item $a^{\parallel k}\parallel a^{\parallel k'} = a^{\parallel (k+k')}$;
    \item $(aa)^{\parallel k}\parallel (aa)^{\parallel k'} = (aa)^{\parallel (k+k')}$;
    \item $((aa)\parallel a)\parallel a^k = a^k$;
    \item $((aa)\parallel a)\parallel (aa)^k = (aa)^k$.
  \end{itemize}
  \end{multicols}
  \noindent with $(aa)\parallel a$ idempotent, and the sequential product $\cdot$ verifying
  \begin{multicols}{2}
  \begin{itemize}
    \item $a\cdot a = aa$;
    \item $((aa)\parallel a)\cdot a = ((aa)\parallel a) a$.
  \end{itemize}
  \end{multicols}
  \noindent such that $0$ is a zero for both products, and all products undefined above are mapped to $0$.
  It recognizes $L$ since $L=\varphi_{\sim_{\mathcal{A}}}^{-1}(((aa)\parallel a) a)$.
  We have $\varphi_{\sim_{\mathcal{A}}}(Seq(SP^+(A)))=\{a,aa,((aa)\parallel a) a,0\}$, whose cardinality is 4.
  Consider $\mathbb{N}^{4*}$, and define the morphism of commutative semigroups $\psi:(\mathbb{N}^{4*},+)\to (SP^+(A)/\mathord\sim_\mathcal{A},\parallel)$ by 
  \begin{multicols}{2}
    \begin{itemize}
    \item $\psi((1,0,0,0))=a$;
    \item $\psi((0,1,0,0))=aa$;
    \item $\psi((0,0,1,0))=((aa)\parallel a)a$;
    \item $\psi((0,0,0,1))=0$.
    \end{itemize}
  \end{multicols}

  Equip $\mathbb{N}^{4*}$ with its sequential product. We have $n(1,0,0,0)=(0,0,1,0)$ for all $n\in\{(k,k,0,0) : k>0 \}$, $(1,0,0,0)(1,0,0,0)=(0,1,0,0)$ and all other sequential products are mapped to $(0,0,0,1)$.
  Define also the morphism of sp-algebras $\mu:SP^+(A)\to \mathbb{N}^{4*}$ by $\mu(a)=\psi^{-1}\varphi(a)$ for all $a\in A$. 
  Note that $L=\mu^{-1}((0,0,1,0))$ and $SP^+(A)-L=\mu^{-1}(\mathbb{N}^{4*}-\{(0,0,1,0)\})$ with $\mathbb{N}^{4*}-\{(0,0,1,0)\}$ a $\parallel$-rational language of $\mathbb{N}^{4*}$, since
  \begin{align*}
    \mathbb{N}^{4*}-\{(0,0,1,0)\} =& (1,0,0,0)\parallel B^\circledast + (0,1,0,0)\parallel B^\circledast + (0,0,1,0)\parallel B^\oplus + (0,0,0,1)\parallel B^\circledast
  \end{align*}
  where $B=\{(1,0,0,0),(0,1,0,0),(0,0,1,0),(0,0,0,1)\}$.

  Set 
  \begin{multicols}{2}
  \begin{itemize}
  \item $D_1=\{(2,4),(4,5),(3,6),(7,8)\}$;
  \item $D_2=\{(2,5)\}$;
  \item $D_3=\{(1,8)\}$;
  \item $D_4=\{(1,7)\}$.
  \end{itemize}
  \end{multicols}
  We have
  \begin{multicols}{2}
  \begin{itemize}
  \item $\Delta_{D_1}^{\mathbb{N}^{k*}}=\{(1,0,0,0)\}$;
  \item $\Delta_{D_2}^{\mathbb{N}^{k*}}=\{(0,1,0,0)\}$;
  \item $\Delta_{D_3}^{\mathbb{N}^{k*}}=\{(0,0,1,0)\}$;
  \item $\Delta_{D_4}^{\mathbb{N}^{k*}}=\{(k,k,0,0) : k>0 \}$;
  \item $\Delta_D^{\mathbb{N}^{k*}}=\emptyset$ for all $D\in\mathcal{P}^+(Q^2)-\{D_1,D_2,D_3,D_4\}$;
  \item $\Delta_\emptyset^{\mathbb{N}^{k*}}=\mathbb{N}^{4*}-(\cup_{i\in[4]}\Delta_{D_i}^{\mathbb{N}^{k*}})$.
  \end{itemize}
  \end{multicols}

  Now, from any element $s\in\mathbb{N}^{4*}$, say for example $s=(1,2,1,0)$, we construct on automaton $\mathcal{A}_s$ such that $L(\mathcal{A}_s)=\mu^{-1}(s)$, following the construction of the proof of Lemma~\ref{lem:revImageMuRatElt}. Note that $\mu^{-1}((1,2,1,0))=\{a\parallel(aa)\parallel(aa)\parallel((a\parallel(aa))^{\parallel k}a) : k>0 \}$. 
  The first step of the construction consists in writing all the sets $\Delta_D^{\mathbb{N}^{k*}}$, $D\in\mathcal{P}(Q^2)$, as a union of disjoint linear sets of $\mathbb{N}^{4*}$. This is trivial when $D$ is one of $D_1,D_2,D_3$, or $D\in\mathcal{P}^+(Q^2)-\{D_1,D_2,D_3,D_4\}$. For $\Delta_{D_4}^{\mathbb{N}^{k*}}$, we have $\Delta_{D_4}^{\mathbb{N}^{k*}}=(1,1,0,0)\parallel (1,1,0,0)^\circledast$. This could also easily be done for $\Delta_\emptyset^{\mathbb{N}^{k*}}$, but it can be avoided. % for the particular case of the construction of $\mathcal{A}_{(1,2,1,0)}$. 
Indeed, assume that $\Delta_\emptyset^{\mathbb{N}^{k*}}$ is partitioned into finitely many linear sets: $\Delta_\emptyset^{\mathbb{N}^{k*}}=\cup_{i\in I_\emptyset}\Delta_{\emptyset,i}^{\mathbb{N}^{k*}}$ with $I_\emptyset$ a finite set.
Assume also that $\mathcal{A}_s$ has one state $\Delta_{\emptyset,i}$ for each $i\in I_\emptyset$. Take one of those states $\Delta_{\emptyset,i}$. 
Following the construction of $\mathcal{A}_s$, it can be easily checked that if a path uses one of the states  $\Delta_{\emptyset,i}$ or $(0,0,0,1)$ then it continues either in $\Delta_{\emptyset,i}$ or in  $(0,0,0,1)$: the final state $s$ of $\mathcal{A}_s$ is unreachable.

  Let us then return to the construction of $\mathcal{A}_s$. Choosing $m$ as small as possible using the remark above, we have $m=(2,2,1,0)$, then $S_1=\{(x_1,x_2,x_3,0)\in\mathbb{N}^{4} : x_1,x_2\leq 2, x_3\leq 1\}$, and we can reduce $S_2$ to $S_2=\{\Delta_{D_4}\}$. The initial and final states are respectively $1$ and $s$. The (useful) sequential transitions are 
\begin{multicols}{2}
\begin{itemize}
\item $(1,a,(1,0,0,0))$, 
\item $((1,0,0,0),a,(0,1,0,0))$,
\item $((k,k,0,0),a,(0,0,1,0))$ for all $0<k\leq 2$,
\item $(\Delta_{D_4},a,(0,0,1,0))$.
\end{itemize}
\end{multicols}
\noindent The fork transitions are from all state $t$ to $\{1,1,\overline{t}_1\}$. 
Finally, the join transitions are 
\begin{itemize}
\item $(\{\overline{1}_2,(x_1,x_2,x_3,0),(x'_1,x'_2,x'_3,0)\},(x_1+x'_1,x_2+x'_2,x_3+x'_3,0))$ when $x_1+x'_1, x_2+x'_2\leq 2$ and $x_3+x'_3\leq 1$,
\item $(\{\overline{1}_2,(x_1,x_2,0,0),(x'_1,x'_2,0,0)\},\Delta_{D_4})$ when $x_1+x'_1=x_2+x'_2>2$,
\item $(\{\overline{1}_2,\Delta_{D_4},(1,1,0,0)\},\Delta_{D_4})$.
\end{itemize}

Figure~\ref{fig:pathAs} represents a successful path in $\mathcal{A}_s$ labeled by $a\parallel(aa)\parallel(aa)\parallel((a\parallel(aa))^{\parallel 3}a)$.
  \begin{figure}[htbp]
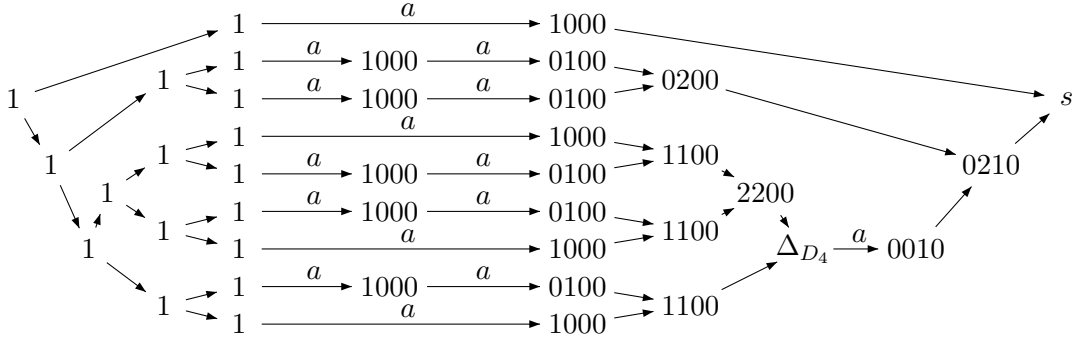

    \begin{center}
      \gasset{Nw=6,Nh=6,loopdiam=4}
      \begin{gpicture}
        \node[Nframe=n](30)(70,20){$1$}
        \node[Nw=10,Nframe=n](31)(115,20){$1000$}
        \drawedge(30,31){$a$}
        \node[Nframe=n](51)(50,-10){$1$}
        \node[Nframe=n](52)(60,12.5){$1$}
        \node[Nframe=n](53)(70,15){$1$}
        \node[Nframe=n](54)(60,2.5){$1$}
        \node[Nframe=n](55)(70,10){$1$}
        \node[Nframe=n](56)(70,5){$1$}
        \node[Nframe=n](57)(70,0){$1$}
        \drawedge(54,56){}
        \drawedge(54,57){}
        \node[Nw=10,Nframe=n](58)(90,10){$1000$}
        \node[Nw=10,Nframe=n](60)(90,0){$1000$}
        \node[Nw=10,Nframe=n](61)(130,2.5){$1100$}
        \node[Nw=10,Nframe=n](62)(130,12.5){$0200$}
        \node[Nw=10,Nframe=n](63)(90,15){$1000$}
        \node[Nw=8,Nframe=n](64)(145,-10){$\Delta_{D_4}$}
        \node[Nw=10,Nframe=n](65)(115,10){$0100$}
        \node[Nw=10,Nframe=n](66)(115,5){$1000$}
        \node[Nw=10,Nframe=n](67)(115,0){$0100$}
        \node[Nw=10,Nframe=n](68)(115,15){$0100$}
        \drawedge(52,53){}
        \drawedge(52,55){}
        \drawedge(53,63){$a$}
        \drawedge(57,60){$a$}
        \drawedge(55,58){$a$}
        \drawedge(56,66){$a$}
        \drawedge(58,65){$a$}
        \drawedge(60,67){$a$}
        \drawedge(63,68){$a$}
        \drawedge(65,62){}
        \drawedge(68,62){}
        \drawedge(66,61){}
        \drawedge(67,61){}
        \node[Nframe=n](7)(52.5,-2.5){$1$}
        \node[Nframe=n](8)(60,{-7.5}){$1$}
        \node[Nframe=n](9)(70,{-5}){$1$}
        \node[Nframe=n](10)(60,{-17.5}){$1$}
        \drawedge(7,8){}
        \drawedge(7,54){}
        \node[Nframe=n](11)(70,{-10}){$1$}
        \node[Nframe=n](12)(70,{-15}){$1$}
        \node[Nframe=n](13)(70,{-20}){$1$}
        \drawedge(10,12){}
        \drawedge(10,13){}
        \node[Nw=10,Nframe=n](15)(90,{-15}){$1000$}
        \node[Nw=10,Nframe=n](17)(130,{-17.5}){$1100$}
        \node[Nw=10,Nframe=n](18)(130,{-7.5}){$1100$}
        \node[Nw=10,Nframe=n](19)(90,{-5}){$1000$}
        \node[Nw=10,Nframe=n](20)(140,{-2.5}){$2200$}
        \node[Nw=10,Nframe=n](25)(115,-10){$1000$}
        \node[Nw=10,Nframe=n](26)(115,-15){$0100$}
        \node[Nw=10,Nframe=n](27)(115,-20){$1000$}
        \node[Nw=10,Nframe=n](28)(115,-5){$0100$}
        \drawedge(18,20){}
        \drawedge(8,9){}
        \drawedge(8,11){}
        \drawedge(9,19){$a$}
        \drawedge(13,27){$a$}
        \drawedge(11,25){$a$}
        \drawedge(12,15){$a$}
        \drawedge(15,26){$a$}
        \drawedge(19,28){$a$}
        \drawedge(25,18){}
        \drawedge(28,18){}
        \drawedge(26,17){}
        \drawedge(27,17){}
        \drawedge(51,7){}
        \drawedge(51,10){}
        \node[Nframe=n](71)(45,1.25){$1$}
        \drawedge(71,52){}
        \drawedge(71,51){}
        \node[Nframe=n](72)(40,10){$1$}
        \drawedge(72,30){}
        \drawedge(72,71){}
        \drawedge(61,20){}
        \drawedge(20,64){}
        \drawedge(17,64){}
        \node[Nw=10,Nframe=n](73)(160,-10){$0010$}
        \drawedge(64,73){$a$}
        \node[Nw=10,Nframe=n](74)(170,1.25){$0210$}
        \drawedge(73,74){}
        \drawedge(62,74){}
        \node[Nframe=n](75)(180,10){$s$}
        \drawedge(74,75){}
        \drawedge(31,75){}
      \end{gpicture}
      \caption{A path labeled by $a\parallel(aa)\parallel(aa)\parallel((a\parallel(aa))^{\parallel 3}a)$ in $\mathcal{A}_s$. In order to lighten the picture, states of the form $(x_1,x_2,x_3,x_4)$ are denoted $x_1x_2x_3x_4$. Also, the special states and transitions using them have been removed or simplified.}
      \label{fig:pathAs}
    \end{center}
  \end{figure}
\end{exa}

\section{P-MSO}
\label{sec:PMSO}

In this section we define a logical formalism called P-MSO, which is a mix between Presburger~\cite{Presburger1930} and monadic second-order logic, and that has exactly the same expressivity as branching automata. As all the constructions involved in the proof are effective, then the P-MSO theory of the class of finite N-free posets is decidable.

Let us recall useful elements of monadic second-order logic, and settle some notation. For more details about MSO logic we refer e.g. to Thomas' survey paper~\cite{EF99,Thomas97a}. The monadic second-order (MSO) logic is classical in set theory, and was first set up by B\"uchi-Elgot-Trakhtenbrot for words~\cite{Buc60,Elgot61,Trakhtenbrot62}. In our case, the domain of interpretation is the class of finite N-free posets.

Monadic second-order logic is an extension of first-order logic  that allows to quantify over elements as well as subsets of the domain of the structure. Given a signature $\mathcal L$, one can define the set of {\em MSO-formul{\ae}}\ over ${\mathcal L}$ as well-formed formul{\ae}\ that can use first-order variable symbols $x,y,\dots$ interpreted as elements of the domain of the structure, monadic second-order variable symbols $X,Y,\dots$ interpreted as subsets of the domain, symbols from ${\mathcal L}$, and a new unary predicate $X(x)$, also denoted $x\in X$ for readability, interpreted as ``the interpretation of $x$ belongs to the interpretation of $X$''. We call \emph{sentence} any formula without free variable. As usual, we will often confuse logical symbols with their interpretation.

Given a signature ${\mathcal L}$ and an  ${\mathcal L}-$structure $M$ with domain $D$, we say that a relation $R \subseteq D^m \times \Pf{D}^n$ is \emph{MSO-definable} in $M$ if and only if there exists an MSO-formula over ${\mathcal L}$, say $\psi(x_1,\dots,x_m,X_1,\dots,X_n)$, which is true in $M$ if and only if $(x_1,\dots,x_m,X_1,\dots,X_n)$ is interpreted by an $(m+n)$-tuple of $R$. 

Given a finite alphabet $A$, let us consider the signature ${\mathcal L}_A=\{<,(R_a)_{a \in A}\}$ where $<$ is a binary relation symbol and the $R_a$'s are unary predicates (over first-order variables). One can associate to every poset $(P,<,\rho)$ labeled over $A$ the ${\mathcal L}_A-$structure $M_{(P,<,\rho)}=(P;<;(R_a)_{a \in A})$ where $<$ is interpreted as the ordering over $P$, and $R_a(x)$ holds if and only if $\rho(x)=a$. In order to take into account the case $P=\emptyset$, which leads to the structure $M_{\emptyset}$ which has an empty domain, we will allow structures to be empty.  Given an MSO sentence $\psi$ over the signature ${\mathcal L}_A$, we define the language $L_\psi$ as the class of posets $(P,<,\rho)$ labeled over $A$ that satisfy $\psi$, or, using formal notation, such that $M_{(P,<,\rho)} \models \psi$. Two formul{\ae}\ $\psi$ and $\psi'$ are (logically) \emph{equivalent}, denoted by $\psi\equiv\psi'$, if $L_\psi=L_{\psi'}$. We will say that a language $L\subseteq SP^+(A)$ is {\em definable in MSO logic} (or {\em MSO-definable}) if and only if there exists an MSO-sentence $\psi$ over the signature ${\mathcal L}_A$ such that $L=L_\psi$.

In order to enhance readability of formul\ae\ we use several notations and abbreviations for properties expressible in MSO.
The shortcut $a(x)$ is used instead of $R_a(x)$.
The following are usual and self-understanding:
$\phi\rightarrow\psi$, 
$X\subseteq Y$,
$x=y$.
An existential (resp. universal) quantification $\exists x \psi(X)$ (resp. $\forall x\psi(X)$) is \emph{relative to $X$} if  $\exists x \psi(X)\equiv \exists x\ x\in X\land \psi(X)$ (resp. $\forall x \psi(X)\equiv \forall x\ x\in X\rightarrow\psi(X)$). Relative existential (resp. universal) quantification of $x$ over $X$ is denoted $\exists^X x$ (resp. $\forall^X x$). The notion of relative quantification naturally extends to second-order variables.

MSO logic is strictly less expressive than automata. There is no MSO-formula that defines the language $(a\parallel b)^\oplus$. On the contrary, 
MSO-definability implies rationality.

In order to capture the expressiveness of automata with logic we need to add Presburger expressivity to MSO.
Presburger logic is the first-order logic over the structure $(\mathbb{N},+)$ where $+=\{(a,b,c) : a+b=c\}$.
A language $L\subseteq \mathbb{N}^n$ is a \emph{Presburger set} of $\mathbb{N}^n$ if $L=\{(x_1,\dots,x_n) : \varphi(x_1,\dots,x_n) \text{ is true }\}$ for some Presburger formula  $\varphi(x_1,\dots,x_n)$. If $\varphi(x_1,\dots,x_n)$ is given then $L$ is called the \emph{Presburger set} of $\varphi(x_1,\dots,x_n)$ (or of $\varphi$ for short).
Presburger logic provides tools to manipulate semi-linear sets of $A^\circledast$ with formul{\ae}.
Indeed, let $A=\{a_1,\dots,a_n\}$ be an alphabet.
As a word $u$ of $A^\circledast$ can be thought of as a $n$-tuple $(\vert u\vert_{a_1},\dots,\vert u\vert_{a_n})$ of non-negative integers, where $\vert u\vert_a$ denotes the number of occurrences of letter $a$ in $u$, then $A^\circledast$ is isomorphic to $\mathbb{N}^n$.

\begin{exa}
  \label{ex:pres1}
  Let $A=\{a,b,c\}$ and $L=\{ u\in A^\circledast : \vert u\vert_a\leq \vert u\vert_b\leq \vert u\vert_c \}$.
  Then $L$ is isomorphic to $\{ (n_a,n_b,n_c)\in\mathbb{N}^3 : n_a\leq n_b\leq n_c \}$, and thus the Presburger set of $$\varphi(n_a,n_b,n_c)\equiv(\exists x\ n_b=n_a+x)\land (\exists y\ n_c=n_b+y)$$
\end{exa}

Semi-linear sets and Presburger sets are connected by the following Theorem:
\begin{thm}[Ginsburg and Spanier, Theorem~1.3 of~\cite{GS:PJM66}] % Th 1.3 de \cite{GS:PJM66}
  \label{th:PresburgerSemilinear}
  Let $A=\{a_1,\dots,a_n\}$ be an alphabet and $L\subseteq A^\circledast$.
  Then $L$ is semi-linear if and only if it is the Presburger set of some Presburger formula $\varphi(x_1,\dots,x_n)$. 
  Furthermore, the construction of one description from the other is effective.
\end{thm}

The \emph{P-MSO logic} is a melt of Presburger and MSO logics.
From the syntactic point of view, P-MSO logic contains MSO logic, and in addition formul{\ae}\ of the form 
$$\mathcal{Q}_X(Z,(\psi_1(X),x_1),\dots,(\psi_n(X),x_n),\varphi(x_1,\dots,x_n))$$ 
where $X$ is the name of a new second-order variable, $Z$ is the name of a (free) second-order variable, $\psi_i(X)$ (for each $i\in[n]$) a P-MSO formula having 
no free first-order variables%
, and only quantifications relative to $X$, and $\varphi(x_1,\dots,x_n)$ a Presburger formula with $n$ free variables $x_1,\dots,x_n$. 
As the variable $X$ is for the internal use of $\mathcal{Q}_X$, then it is bounded by $\mathcal{Q}_X$: it is a free variable of all the $\psi_i(X)$, $i\in[n]$, but it is not a free variable of $\psi(Z)=\mathcal{Q}_X(Z,(\psi_1(X),x_1),\dots,(\psi_n(X),x_n),\varphi(x_1,\dots,x_n))$. Similarly, $x_1,\dots,x_n$ are free variables of $\varphi(x_1,\dots,x_n)$, but they must not be considered as free in $\psi(Z)$. 

As in monadic second-order logic, the class of syntactically correct P-MSO formul{\ae}\ is closed under boolean operations, and existential and universal quantification over first and second-order variables of a P-MSO formula that are interpreted over elements or sets of elements of the domain of the structure.
Semantics of P-MSO formul{\ae}\ is defined by extension of semantics of Presburger and MSO logics.
The notions of a language  and definability naturally extend from MSO to P-MSO.
Let us turn to the semantics of $$\mathcal{Q}_X(Z,(\psi_1(X),x_1),\dots,(\psi_n(X),x_n),\varphi(x_1,\dots,x_n))$$
A \emph{block} $B$ of a poset $(P,<)$ is a nonempty subset of $P$ such that, if $b,b'\in B$ such that $b<b'$, then for all elements of $p\in P$, if $b\leq p\leq b'$ then $p\in B$. 
A subset $G$ of $P$ is \emph{good} if, for all $p\in P$, if $p$ is comparable to an element of $G$ and incomparable to another, then $p\in G$. 
A \emph{connected block} $C$ of a block $X$ of a poset is a block such that, for any different and incomparable $c,c'\in C$ there exists $c''\in C$ such that $c,c'\leq c''$ or $c''\leq c,c'$.
 
Before continuing with formal definitions, let us give some intuition on the meaning of $\mathcal{Q}_X(Z,(\psi_1(X),x_1),\dots,(\psi_n(X),x_n),\varphi(x_1,\dots,x_n))$. Let $Y$ be an interpretation of a second-order variable $Z$ in $P$, such that $Y$ is a good block of $P$. That means, $Y$ is the poset associated with a sub-term of a term on $A$ describing $P$, and is the parallel composition of $m\geq 1$ connected blocks: $Y=Y_1\parallel\cdots\parallel Y_m$. Take $n$ different colors $c_1,\dots,c_n$. To each $Y_i$ we associate a color $c_j$ with the condition that $Y_i$ satisfies $\psi_j(Y_i)$. Observe that this coloring may not be unique, and may not exist. Denote by $x_j$ the number of uses of $c_j$ in the coloring of $Y$. Then $P,Y\models \mathcal{Q}_X(Z,(\psi_1(X),x_1),\dots,(\psi_n(X),x_n),\varphi(x_1,\dots,x_n))$ if there exists such a coloring with $x_1,\dots,x_n$ satisfying the Presburger condition $\varphi(x_1,\dots,x_n)$.

More formally,
 let $P\in SP^+(A)$, $\mathcal{Q}_X(Z,(\psi_1(X),x_1),\dots,(\psi_n(X),x_n),\varphi(x_1,\dots,x_n))$ be a P-MSO formula, $Y\subseteq P$ be an interpretation of $Z$ in $P$ such that $Y$ is a good block of $P$. Then $P,Y\models \mathcal{Q}_X(Z,(\psi_1(X),x_1),\dots,(\psi_n(X),x_n),\varphi(x_1,\dots,x_n))$ if there exist non negative integers $v_1,\dots,v_n$ and a partition $(Y_{1,1},\dots,Y_{1,v_1},\dots,Y_{n,1},\dots,Y_{n,v_n})$ of $Y$ into connected blocks $Y_{i,j}$ such that
 \begin{itemize}
 \item $(v_1,\dots,v_n)$ belongs to the Presburger set of $\varphi(x_1,\dots,x_n)$,
 \item $y\in Y_{i,j}$, $y'\in Y_{i',j'}$ implies that $y$ and $y'$ are incomparable, for all possible $(i,j)$ and $(i',j')$ with $(i,j)\not=(i',j')$,
 \item $P,Y_{i,j}\models \psi_i(Y_{i,j})$ for all $i\in[n]$ and $j\in[v_i]$.
 \end{itemize}

 \begin{exa}
   \label{ex:pmso1}
   Let $L$ be the language of Example~\ref{ex:pres1}, and $\varphi(n_a,n_b,n_c)$ be the Presburger formula of Example~\ref{ex:pres1}.
   For all $\alpha\in A$, set
   $\psi_\alpha(X)\equiv\texttt{Card}_1(X)\land\forall^X x\ \alpha(x)$,
   where $\texttt{Card}_1(X)$ is a MSO formula (thus a P-MSO formula) which is true if and only if the interpretation of $X$ has cardinality 1. 
   Then $L$ is the language of the following P-MSO sentence:
     $$\forall P\ (\forall p\ p\in P)\rightarrow \mathcal{Q}_X(P,(\psi_a(X),n_a),(\psi_b(X),n_b),(\psi_c(X),n_c),\varphi(n_a,n_b,n_c))$$
 \end{exa}

\begin{thm}
  \label{th:PMSORat}
  Let $A$ be an alphabet, and $L\subseteq SP^+(A)$.
  Then $L$ is rational if and only if is P-MSO definable.
  Furthermore the construction from one formalism to the other is effective.
\end{thm}

The remainder of this Section is devoted to the proof of Theorem~\ref{th:PMSORat}.

\subsection{From automata to P-MSO}

The inclusion from left to right of Theorem~\ref{th:PMSORat} relies on the ideas of B\"uchi on words: the encoding of accepting paths of a branching automaton $\mathcal{A}$ into a P-MSO formula.
More precisely, for words this part of the proof consists in mapping each letter of the word to a state of the automaton, consistently with the transitions. In our case, each letter of the poset is mapped to a sequential transition of $\mathcal{A}$, and each part of the poset of the form $P=P_1\parallel\cdots\parallel P_n$ ($n>1)$, as great as possible relatively to inclusion and such that each $P_i$ is a connected block of $P$, is mapped to a pair $(p,q)$ of states; informally speaking, $p$ and $q$ are the states that are supposed to respectively begin and finish the part of the path labeled by $P$. The formula guarantees that pairs of states and sequential transitions are chosen consistently with the transitions of $\mathcal{A}$, and that, if $P=P_1\parallel\cdots\parallel P_n$ as above and $p_i \mathop{\Longrightarrow}\limits_{\mathcal{A}}^{P_i} q_i$ for all $i\in[n]$, then there exists a combination of fork transitions that connects $p$ to $p_1,\dots,p_n$, a sequence of join transitions that connects $q_1,\dots,q_n$ to $q$, such that a path $p \mathop{\Longrightarrow}\limits_{\mathcal{A}}^{P} q$ in $\mathcal{A}$ is formed.

Let us give this construction more formally.

Given two second-order variables representing sets, the  properties
$X\subset Y$,  $X\subseteq Y$, $X\cap Y=\emptyset$
are clearly definable in MSO.
The equality $x=y$ of two elements is clearly expressible with a MSO formula, that states for example that $\{x\}\subseteq\{y\}\land \{y\}\subseteq\{x\}$.
We denote 
``there exists an unique $x$'' by $\exists!x$, 
``$x$ and $y$ are different and not comparable'' by $x\parallel y$,
``there exists a non-empty set $X$'' by $\overline{\exists} X$,
``set $X$ has cardinality $j$'' by $\texttt{Card}_j(X)$, where $j$ is any integer,
``set $X$ has cardinality $>j$'' by $\texttt{Card}_{>j}(X)$, where $j$ is any integer,
``there exists (resp. for all) $x$ in $X$'' by $\exists^X x$ (resp. $\forall^X x$),
``$X$ contains all the elements'' by $\texttt{Universe}(X)$.
All those properties are definable in MSO.
In the further we will use the following shortcuts:
\begin{align*}
  y<X\equiv& \forall x\ x\in X\rightarrow y<x\\
  y\parallel X\equiv& (\lnot y\in X)\land(\forall x\ x\in X\rightarrow y\parallel x)\\
  X\parallel Y\equiv& X\cap Y=\emptyset\land\forall x\forall y (x\in X\land y\in Y)\rightarrow x\parallel y\\
  \texttt{Pred}(X,Y)\equiv&\forall x\ x\in X\rightarrow\exists y\ y\in Y\land x<y\land\lnot\exists z\ x<z\land z<y\\
  &\land\forall y\forall x( y\in Y\land x<y\land\lnot\exists z\ x<z<y )\rightarrow x\in X\\
  \texttt{Antichain}(X)\equiv&\forall^X x\forall^X x'\ \lnot(x<x'\lor x'<x)\\
  \texttt{Min}(M,X)\equiv&M\subseteq X\land \texttt{Antichain}(M)\land\forall^X x\exists^M m\ m\leq x                 
\end{align*}

\noindent The formal definitions of a block, a good block and a connected block can be directly translated into MSO formul{\ae}\ as follows:
 \begin{align*}
   \texttt{ConnectedBlock}(C,X)\equiv& \texttt{Block}(C,X)\land(\forall^C c\forall^C c'\ (c\not=c'\land\texttt{Incomp}(c,c')) \rightarrow\\&\ \hskip4cm\exists^C c''\ \texttt{Comp}(c,c'')\land\texttt{Comp}(c',c''))\\    
   \texttt{GoodBlock}(R,X)\equiv&
   \texttt{Block}(R,X)\land
   (\forall^R r\forall^R r'\forall^X x (\texttt{Comp}(x,r)\land\texttt{Incomp}(x,r'))\rightarrow R(x))\\
   \texttt{Comp}(x,y)\equiv& x<y\lor y<x\\
   \texttt{Incomp}(x,y)\equiv& (\lnot x<y)\land(\lnot y<x)\\
   \texttt{Block}(R,X)\equiv&R\subseteq X\land\texttt{Card}_{>0}(R)\land(\forall^R r\forall^X x\forall^R r'\ r<x\land x<r'\rightarrow R(x))
 \end{align*}\smallskip

\noindent Let $A$ be an alphabet and $L$ be a rational language of $SP^+(A)$ given by an automaton $\mathcal{A}=(Q,A,E,I,F)$ and $P\in SP^+(A)$.
Following B\"uchi's ideas~\cite{Buc60}, this section is devoted to the construction of a P-MSO sentence $\phi_\mathcal{A}$ such that $i \mathop{\Longrightarrow}\limits_{\mathcal{A}}^{P} f$ for some $i\in I$, $f\in F$, if and only if $P\models\phi_\mathcal{A}$. 

Informally speaking, the sentence annotates $P$ using second-order variables that encode the transitions of $\mathcal{A}$, consistently with the notion of a path.
To each transition $(p,a,q)$ we attach a second-order variable $X_{(p,a,q)}$, and use the following formula
\begin{multline*}
  \texttt{MarkSeq}\equiv \forall x\mathop\land_{a\in A}(a(x)\rightarrow\mathop\lor_{(p,a,q)\in E}X_{(p,a,q)}(x))\\
  \land(\forall x\mathop\land_{(p,a,q)\in E}(X_{(p,a,q)}(x)\rightarrow\lnot\mathop\lor_{a'\in A\atop {(p',a',q')\in E\atop(p',a',q')\not=(p,a,q)}}X_{(p',a',q')})))
\end{multline*}
in order to express that each element of $P$ labeled by $a$ is the label of a unique transition $(p,a,q)$ in a path. 
In order to express that if $(p,a,q)$ and $(p',b,q')$ are two consecutive transitions in a path then $q=p'$, we use the following formula
\begin{multline*}
  \texttt{ConsistentSeq}\equiv \forall x\forall y(\texttt{Succ}(\{y\},\{x\})\land \texttt{Pred}(\{x\},\{y\}))\\
  \rightarrow \mathop\land_{(p,a,q)\in E}(X_{(p,a,q)}(x)\rightarrow\mathop\lor_{(q,b,r)\in E}X_{(q,b,r)}(y))
\end{multline*}
We now turn to a more technical part of the construction of $\phi_\mathcal{A}$: expressing that each part of $P$ of the form $R_1\parallel\dots\parallel R_n$, $n>1$, is the label of path from a state $s$ to a state $t$ that uses fork and join transitions consistently to the definition of a path in $\mathcal{A}$. 
Second-order variables $X^-_{s,t}$ and $X^+_{s,t}$ are used to express that the part of the path labeled by $R$ starts in $s$ and finishes in $t$. The sets $X^-_{s,t}$ and $X^+_{s,t}$ are composed of particular elements of $P$ which are named \emph{witnesses} of $R$, and are attached unambiguously to $R$.

A \emph{good maximal parallel block} of $X\subseteq P$ is a good block $R$ of $X$ such that $R=P_1\parallel P_2$ for some nonempty $P_1$ and $P_2$, and maximal relatively to parallel decomposition, 
ie. there is no good block $R'$ of $X$ such that $R'=R\parallel P_3$ for some nonempty $P_3$. 
The property ``$R$ is a good maximal parallel block of $X$'' can easily be translated into a formula $\texttt{GMPB}(R,X)$. 
\begin{multline*}
  \texttt{GMPB}(R,X)\equiv
  R\subseteq X \land
  \texttt{GoodBlock}(R,X) \land
  (\overline{\exists} R_1\overline{\exists} R_2\ R=R_1\parallel R_2)\land\\
  \forall R'\ (R'\subseteq X\land \texttt{GoodBlock}(R',X) \land \overline{\exists} R_1\overline{\exists} R_2\ R'=R_1\parallel R_2)\rightarrow \lnot\overline{\exists} T\ R'=R\parallel T
\end{multline*}

\begin{exa}
  \label{ex:gmpb}
  Figure~\ref{fig:gmbp} represents a N-free poset $P$ and its good maximal parallel blocks $G_1$, $G_2$, $G_3$ and $G_4$.
  \begin{figure}
      \hskip-6cm\input{gmpb.tex.pdf}% Pour pdflatex
      \hskip-6.7cm\includegraphics{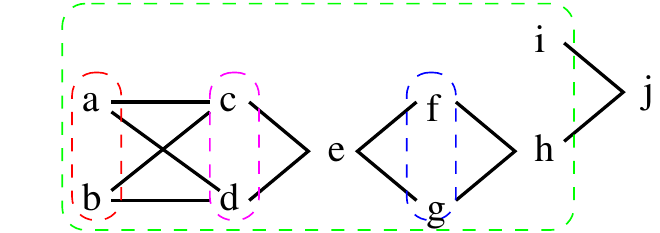}% Pour pdflatex
      
      \vskip-3cm\hskip6cm\synttree[$\cdot$%
    [{{$\color{green}\parallel\ (G_1)$}}%
      [$\cdot$%
        [{{$\color{red}\parallel\ (G_2)$}}[$a$][$b$]]%
        [{{$\color{magenta}\parallel\ (G_3)$}}[$c$][$d$]]%
        [$e$]%
        [{{$\color{blue}\parallel\ (G_4)$}}[$f$][$g$]]%
        [$h$]
      ]%
      [$i$]%
    ]%
    [$j$]%
  ]%
  \caption{A N-free poset and its good maximal parallel blocks.}
  \label{fig:gmbp}
  \end{figure}
  The good maximal parallel block $G_1$ can be decomposed into $G_1=C_1\parallel C_2$ where $C_1=\{i\}$ and $C_2=G_1-C_1$ are two connected blocks.
  The good maximal parallel block $G_4$ has two connected blocks: $\{f\}$ and $\{g\}$.
  Any N-free poset $P$ can be represented by a labeled tree, as in the figure, where internal nodes are labeled by sequential or parallel products, and leaves by elements of $P$, and such that no internal node has the same label as one of its sons. Formally, these trees are not terms as we defined them because of the arity of internal nodes. Good maximal parallel blocks of $P$ correspond to the sub-trees whose root is a node labeled by $\parallel$.
\end{exa}

The following Lemma holds.
\begin{lem}
  \label{lem:compGMPB}
  Let $P$ be a N-free poset, $X\subseteq P$, and $R,R'$ be two good maximal parallel blocks of $X$.
  Then either $R<R'$ or $R'<R$ or $R\parallel R'$ or $R\subseteq R'$ or $R'\subseteq R$.
\end{lem}
\begin{proof}
  First assume $R\cap R'=\emptyset$. Let $R_1,R_2,R'_1,R'_2$ be such that $R=R_1\parallel R_2$ and $R'=R'_1\parallel R'_2$. Assume that there exist $r\in R$, $r'\in R'$ such that $r$ and $r'$ are comparable. Wlog suppose $r<r'$. As $R$ is a good block, then $x<r'$ for all $x\in R$. It follows easily that $R<R'$. If all the elements of $R$ and $R'$ are incomparable then $R\parallel R'$.

  Assume now that there exists $x\in R\cap R'$. Assume that 
one is not included into the other. Thus there exist $r\in R-R'$ and $r'\in R'-R$. Assume that $x$ and $r$ are comparable, say wlog. $r<x$. If $r$ and $r'$ are incomparable then $r\in R'$ because $R'$ is a good block. So necessarily $r<r'$ otherwise $R'$ would not be a good block. If $x$ and $r'$ are incomparable then $R$ is not a good block, so $x$ and $r'$ are comparable and necessarily $x<r'$ otherwise $R$ would not be a block. Now, as $R$ and $R'$ are both parallel blocks, there exist $a\in R$ and $b\in R'$ such that $a$ is incomparable to $r$ and $x$, and $b$ is incomparable to $r'$ and $x$. As $R$ and $R'$ are good blocks then $a<r'$ and $r<b$. Because $P$ is N-free then $a<b$. Thus $\{r,x,a,b\}$ forms an N, which is a contradiction. So $r$ (and $r'$) is uncomparable to $x$. As a consequence $r$ and $r'$ are also incomparable, otherwise $R$ and $R'$ would not be good blocks. As $r$ is incomparable to $x$, we also have that if $R=R_1\parallel\cdots\parallel R_n$ is a decomposition of $R$ and $x\in R_i$ for some $i\in[n]$ then $R_i\subseteq R\cap R'$. Also, decomposing $R'$ into $R'=R'_1\parallel\cdots\parallel R'_m$, if $R_i,R'_j\not\subseteq R\cap R'$, then the elements of $R_i$ are incomparable to those of $R'_j$. Thus, consider $R\cup R'$: it is a parallel block. Assume it is not good. There exist $x\in X-(R\cup R')$, $a,b\in R\cup R'$ such that $x$ is comparable to $a$ and incomparable to $b$. If $a,b\in R$ (resp. $a,b\in R'$) then $R$ (resp. $R'$) is not a good block. So assume $a\in R-R'$ and $b\in R'-R$. Necessarily $r$ is comparable to all the elements of $R$ and incomparable to those of $R'$. Thus $R\cap R'=\emptyset$, which is a contradiction. As a conclusion, $R\cup R'$ is a good parallel block, $R,R'\subset R\cup R'$, so $R$ and $R'$ are not maximal: the intersection of two good maximal parallel blocks is necessarily empty.
\end{proof}

The following Lemma shows that the set of minimum (or maximum) elements of a good maximal parallel block $G$ characterizes $G$:
\begin{lem}
  \label{lem:unicityGMPBMin}
  Two different good maximal parallel blocks of $P$ can not have the same set of minimum (resp. maximum) elements.
\end{lem}

\begin{proof}
  Assume by contradiction that $G$ and $G'$ have the same set $M$ of minimum elements (the proof is the same if $M$ is the set of maximum elements).
  According to Lemma~\ref{lem:compGMPB}, $G\subset G'$ or $G'\subset G$, say wlog. $G\subset G'$.
  There exists $x\in G'-G$.
  Necessarily $x\not\in M$, thus there exists $m_1\in M$ such that $m_1<x$. Observe that $x$ is necessarily greater than all the elements of $M$, otherwise $G'$ would not be a good part of $P$. But this implies that $G'$ can not be decomposed into $G'=G'_1\parallel G'_2$, so $G'$ is not a good maximal parallel block, in contradiction with the hypothesis.
\end{proof}

Let $G$ be a good maximal parallel block of a N-free poset P. The set of \emph{witnesses} of $G$ is the union of the set of \emph{left witnesses} of $G$ and the set of \emph{right witnesses} of $G$, respectively denoted $Wit_L(G)$ and $Wit_R(G)$ and defined by:
\begin{align*}
Wit_L(G) :& \text{the greatests }x\in P\text{ such that }x<G\text{ and there is no }y\in P-G\\&\text{ such that }y\parallel G\text{ and }x<y;\\
Wit_R(G) :& \text{the smallests }x\in P\text{ such that }x>G\text{ and there is no }y\in P-G\\&\text{ such that }y\parallel G\text{ and }x>y.
\end{align*}
or equivalently:
\begin{align*}
  Wit_L(G)=\{ w\in\max\{x\in P : x<G\} : Succ(w)=\min(G)\}\\
  Wit_R(G)=\{ w\in\min\{x\in P : G<x\} : Pred(w)=\max(G)\}
\end{align*}

\begin{exa}
  In Example~\ref{ex:gmpb},
  \begin{multicols}{2}
  \begin{itemize}
  \item $Wit_L(G_1)=\emptyset$ and $Wit_R(G_1)=\{j\}$,
  \item $Wit_L(G_2)=\emptyset$ and $Wit_R(G_2)=\{c,d\}$,
  \item $Wit_L(G_3)=\{a,b\}$ and $Wit_R(G_3)=\{e\}$,
  \item $Wit_L(G_4)=\{e\}$ and $Wit_R(G_4)=\{h\}$.
  \end{itemize}
  \end{multicols}
\end{exa}

\begin{exa}
    Figure~\ref{fig:gmbp2} represents a N-free poset $P$ and its good maximal parallel blocks $G_1$, $G_2$, $G_3$ and $G_4$.
  \begin{figure}
     \hskip-6cm\input{gmpb2.tex.pdf}% Pour pdflatex
     \hskip-6.7cm\includegraphics{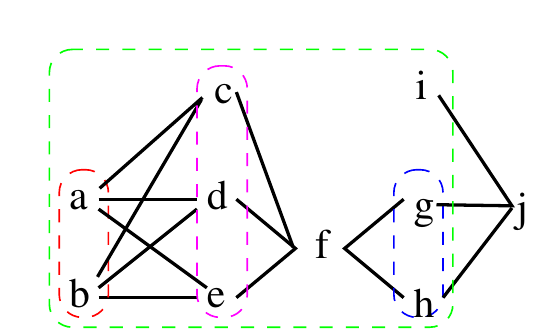}% Pour pdflatex

     \vskip-3cm\hskip6cm\synttree[$\cdot$%
    [{{$\color{green}\parallel\ G_1$}}%
      [$\cdot$%
        [{{$\color{red}\parallel\ G_2$}}[$a$][$b$]]%
        [{{$\color{magenta}\parallel\ G_3$}}[$c$][$d$][$e$]]%
        [$f$]%
        [{{$\color{blue}\parallel\ G_4$}}[$g$][$h$]]%
      ]%
      [$i$]%
    ]%
    [$j$]%
  ]%
  \caption{A N-free poset and its good maximal parallel blocks.}
  \label{fig:gmbp2}
  \end{figure}
  We have
  \begin{multicols}{2}
  \begin{itemize}
  \item $Wit_L(G_1)=\emptyset$ and $Wit_R(G_1)=\{j\}$,
  \item $Wit_L(G_2)=\emptyset$ and $Wit_R(G_2)=\{c,d,e\}$,
  \item $Wit_L(G_3)=\{a,b\}$ and $Wit_R(G_3)=\{f\}$,
  \item $Wit_L(G_4)=\{f\}$ and $Wit_R(G_4)=\emptyset$.
  \end{itemize}
  \end{multicols}
  Observe that $f$ is both a right witness of $G_3$ and a left witness of $G_4$. Observe also that $j$ is not a right witness of $G_4$.
\end{exa}

Let us denote $W^-(G)=Pred(\min G)$, $W^+(G)=Succ(\max G)$ and $W(G)=W^-(G)\cup W^+(G)$.
\begin{lem}
  \label{lem:wit1}
  Either $W^-(G)$ or $W^+(G)$ or $W(G)$ is the set of witnesses of $G$.
\end{lem}

\begin{proof}
  Observe that $Wit_L(G)\subseteq W^-(G)$ and $Wit_R(G)\subseteq W^+(G)$, thus the witnesses of $G$ are a subset of $W(G)$.
  Observe also that if $x\in W^-(G)$ (resp. $W^+(G)$) is a witness of $G$, then all the elements of $W^-(G)$ (resp. $W^+(G)$) are witnesses of $G$. Indeed, let $G=G_1\parallel\dots\parallel G_n$ with $n>1$ and $x\in W^-(G)$ such that $x$ is a witness of $G$. Then $x$ is a predecessor of a minimum of $G$. As $G$ is a good block, then $x$ is also less than all the minimums of $G$, and $x$ is necessarily a predecessor of all those minimums. As in a N-free poset all the predecessors of an element have the same successors it follows that if $x$ is a witness of $G$ then $W^-(G)$ contains only witnesses of $G$. We argue similarly with the elements of $W^+(G)$.
  Now let $w\in W(G)$.
  Wlog., assume that $w\in W^-(G)$. 
  Assume that $w$ is not a witness of $G$: there exists $r$ such that $r\parallel G$ and $w<r$. Let $R=GB(r,G)$ be the smallest block of elements $r'\in P$ such that 
\begin{itemize}
\item $r\in GB(r,G)$ and
\item if $r'$ is comparable to some $r''\in GB(r,G)$ and $r'\parallel G$ then  $r'\in GB(r,G)$
\end{itemize}
Observe that $GB(r,G)$ is a good block for all $r$ such that $r\parallel G$. Indeed, assume that it is not: there exist $r',r''\in GB(r,G)$, $g\in G$ and $t\in P-GB(r,G)$ such that $t$ is comparable to $g$ and $r'$ and incomparable to $r''$, and it is just verification to check that $P$ is not N-free.

Assume there exists $r'\in R$ such that $w$ is incomparable with $r'$. If $r'<r$ then $w,g,r',r$ form an N for any $g\in G$. If $r'$ and $r$ are incomparable then there exists $r''\in R$ such that either $r,r'<r''$ or $r''<r,r'$, otherwise $R$ would not be as small as possible. If  $r,r'<r''$ then $w,g,r'',r'$ is an N for any $g\in G$. The case $r''<r,r'$ is similar. Thus $w<R$.

  As $G\cup R$ is not a good parallel block, but $G$ and $R$ are, there exists $z\in P-(G\cup R)$ such that $z$ is comparable to all the elements of $G$ and incomparable to all the elements of $R$, or $z$ is incomparable to all the elements of $G$ and comparable to all the elements of $R$. As this latter case would imply that $z\in R$, then only the first case is possible. If $z<G$ and $z\parallel R$ then $z$ and $w$ are different, and $z\parallel w$ is impossible, otherwise there would be an N formed by $w,z,r$ and any element of $G$. It is impossible that $z<w$ because it implies $z<w<r$. It is also impossible that $w<z<G$. Now take $z$ as small as possible such that $G<z$ and $z\parallel R$. By contradiction, assume that $z$ is not a witness of $G$; there exists $p\in P-(R\cup G)$ such that $p<z$ and $G\parallel p$. Let $Q=GB(p,G)$. Then $Q$ is a good block, $Q\parallel G$, $Q\parallel R$ and $Q<z$. As $Q\cup G$ can not be a good parallel block (otherwise $G$ would not be a maximal good parallel block), then exists $t\in P-(Q\cup G)$ such that either
\begin{itemize}
\item $t$ is comparable to an element of $Q$ and incomparable to an element of $G$, and so to all the element of $G$ because $G$ is a good block; this implies $t\in Q$ which is in contradiction;
\item or $t$ is incomparable to an element of $Q$ (and thus to all the elements of $Q$) and comparable to an element of $G$ (and thus to all the elements of $G$). They are two cases. If $t<G$ then $t$ is necessarily incomparable and different to $w$, thus for all $g\in G$, $t,p,w,g$ form an N, which is a contradiction. Otherwise, if $G<t$, then necessarily $t<z$ and $t\parallel R$, and thus $z$ is not as small as possible such that $G<z$ and $z\parallel R$, which is also a contradiction.
\end{itemize}
Thus $z$ is a witness of $G$, and $z\in W^+(G)$.

We use similar arguments to show that if $x\in W^+(G)$ is not a witness of $G$ then $W^-(G)$ is the set of witnesses of $G$.
\end{proof}

\begin{cor}
  Every good maximal parallel block $G$ of $P$ which is not $P$ itself has a witness.
\end{cor}

\begin{proof}
  If $W(G)$ is empty then there is no $p\in P-G$ which is comparable to an element of $G$. Thus $(P-G)\not=\emptyset$ is a good block, and $(P-G)\cup G$ a good parallel block, which contradict the maximality of $G$.
  If $W^-(G)$ is not empty but has no witnesses of $G$, then using the same arguments as in proof of Lemma~\ref{lem:wit1} the witnesses of $G$ are the elements of $W^+(G)$ which is not empty. Similarly, if $W^+(G)$ is not empty but has no witnesses of $G$ then the witnesses of $G$ are the elements of $W^-(G)$ which is not empty.
\end{proof}

\begin{lem}
  \label{lem:Succ1}
  Let $P$ be a N-free poset, and $a,b,c,d\in P$ such that $a$ is a predecessor of $b$, $c$ a predecessor of $d$, $a<d$ and $c<b$.
  Then $a$ is a predecessor of $d$ and $c$ a predecessor of $b$.
\end{lem}

\begin{proof}
  Necessarily $a\parallel c$ and $b\parallel d$.
  By contradiction assume that the statement of the Lemma is false, for example that $a$ is not a predecessor of $d$. Then, there exists $x$ successor of $a$ such that $x<d$.
  Then $x$ is incomparable with $b$ and to $c$, thus $x,a,b,c$ is an N.
\end{proof}

\begin{lem}
  \label{lem:wit2}
  Every $x\in P$ is a left (resp. right) witness of at most one good maximal parallel block of $P$.
\end{lem}

\begin{proof}
  Assume that $x$ is a left witness of two different good maximal parallel blocks $G_1$ and $G_2$.
  Then $x\in W^-(G_1)\cap W^-(G_2)$, and thus there exists $g_1\in\min G_1$, $g_2\in\min G_2$ such that $x$ is a predecessor of both $g_1$ and $g_2$.
  A consequence of the definition of left witness is that $g_1\in G_2$ (and thus $g_1\in\min G_2$) and $g_2\in G_1$ (and thus $g_2\in\min G_1$).
  Now, let $y\in W^-(G_1)$. There exists $g'_1\in \min G_1$ such that $y$ is a predecessor of $g'_1$. According to Lemma~\ref{lem:wit1} $y$ is a left witness of $G_1$, and thus $y<g_1$ and $x<g'_1$. As a consequence of Lemma~\ref{lem:Succ1} $x,y$ are both predecessors of $g_1,g'_1,g_2$. It follows that $W^-(G_1)=W^-(G_2)$, and as a consequence $\min G_1=\min G_2$. 
  Assume that $G_1\not= G_2$, and wlog. that there exists $g_2\in G_2$ such that $g_2\not\in G_1$.
  Then $g_2$ is comparable to a minimum $m$ of $G_2$, which is also a minimum of $G_1$, and, because $G_2$ can be decomposed into $G_2=G_{2,1}\parallel G_{2,2}$, $g_2$ is incomparable to another minimum $m'$ of $G_2$, which is also a minimum of $G_1$. Thus, because $G_1$ is good, we should have $g_2\in G_1$, which is a contradiction.
\end{proof}

Being a left witness of a good maximal parallel block $G$ can easily be encoded into a MSO-formula:
$$
  \texttt{Wit}_L(x,G)\equiv 
  \exists M\exists R\ \texttt{Min}(M,G)\land\texttt{Pred}(R,M)\land R(x) 
$$
A similar formula $\texttt{Wit}_R(x,G)$ can be written for right witnesses.

We now come back to the definition of a P-MSO formula that expresses that the (strict) part of $P$ identified by $X$ is the label of a path in $\mathcal{A}$.
To each good maximal parallel block $G$ of $X$ we attach a unique couple of states $(p,q)\in Q\times Q$ using the witnesses of $G$, with the help of a two second-order variables, $X^-_{p,q}$ for left witnesses and $X^+_{p,q}$ for right witnesses:
\begin{multline*}
  \texttt{MarkPar}\equiv\forall X\forall G\ \texttt{Universe}(X)\land\texttt{GMPB}(G,X)\rightarrow\\
  \mathop\lor_{(p,q)\in Q\times Q}(\forall^X x (\texttt{Wit}_L(x,G)\rightarrow X^-_{p,q}(x))\land (\texttt{Wit}_R(x,G)\rightarrow X^+_{p,q}(x)) )\\
  \land (\forall^X x\mathop\land_{(p,q)\in Q\times Q} ((\texttt{Wit}_L(x,G)\land X^-_{p,q}(x))\rightarrow \lnot\mathop\lor_{(p',q')\in Q\times Q\atop (p',q')\not=(p,q)}X^-_{p',q'}(x)))\\
  \land (\forall^X x\mathop\land_{(p,q)\in Q\times Q} ((\texttt{Wit}_R(x,G)\land X^+_{p,q}(x))\rightarrow \lnot\mathop\lor_{(p',q')\in Q\times Q\atop (p',q')\not=(p,q)}X^+_{p',q'}(x)))
\end{multline*}
and we check that
\begin{itemize}
\item for every block of $P$ of the form $Ga$ with $G$ a good maximal parallel block of $P$ and $a\in A$, if $p \mathop{\Longrightarrow}\limits_{\mathcal{A}}^{G} q$ and $s \mathop{\Longrightarrow}\limits_{\mathcal{A}}^{a} r$ then $q=s$ (and symmetrically for blocks of $P$ of the form $aG$);
\item for every block of $P$ of the form $GG'$ with $G,G'$ good maximal parallel blocks of $P$, if $p \mathop{\Longrightarrow}\limits_{\mathcal{A}}^{G} q$ and $s \mathop{\Longrightarrow}\limits_{\mathcal{A}}^{G'} r$ then $q=s$.
\end{itemize}

The check is done with the formula $\texttt{ConsistentPar}_1$ below.
For convenience, we start by defining a formula $\texttt{GMPBMinStarts}_q(M,G,X)$ for every $q\in Q$ such that $P,X,M,G\models \texttt{GMPBMinStarts}_q(M,G,X)$ if and only if $G$ is the good maximal parallel block of $X$ whose set of minimum elements is $M$ (the uniqueness of $G$ is guaranteed by Lemma~\ref{lem:unicityGMPBMin}) and the part of the path labeled by $G$ starts with $q$.
\begin{multline*}
  \texttt{GMPBMinStarts}_q(M,G,X)\equiv \texttt{GMPB}(G,X)\land \texttt{Min}(M,G)\rightarrow\\
  \forall x ((\texttt{Wit}_L(x,G)\rightarrow \mathop\lor_{r\in Q}X^-_{q,r}(x))\land(\texttt{Wit}_R(x,G))\rightarrow \mathop\lor_{r\in Q}X^+_{q,r}(x)) 
\end{multline*}
Formul{\ae}\ $\texttt{GMPBMinEnds}_q(M,G,X)$, $\texttt{GMPBMaxStarts}_q(M,G,X)$ and $\texttt{GMPBMaxEnds}_q(M,G,X)$ could similarly be written in order to define the good maximal parallel block $G$ according to its set of minimal/maximal elements, and to express that the part of the path labeled by $G$ starts/ends with $q$.

Returning to the definition of $\texttt{ConsistentPar}_1$:
\begin{multline*}
  \texttt{ConsistentPar}_1\equiv\forall X\forall G\forall M^-\forall M^+\forall Y^-\forall Y^+\\
  \texttt{Universe}(X)\land\texttt{GMPB}(G,X)\land\texttt{Min}(M^-,G)\land\texttt{Max}(M^+,G)\land \texttt{Pred}(Y^-,M^-)\land \texttt{Succ}(Y^+,M^+)\rightarrow\\
  \forall x\mathop\land_{(p,q)\in Q\times Q}((\texttt{Wit}_L(x,G)\land X^-_{p,q}(x))\lor(\texttt{Wit}_R(x,G)\land X^+_{p,q}(x)))\rightarrow\\
  ( (\texttt{Card}_1(Y^+)\rightarrow \forall y\ Y^+(y)\rightarrow \mathop\lor_{a\in A\atop{r\in Q\atop (q,a,r)\in E}}X_{q,a,r}(y))\land\\
    \hfill(\texttt{Card}_{>1}(Y^+)\rightarrow \forall G'\ \texttt{GMPBMinStarts}_q(Y^+,G',X))  )\\
  \land\hskip6cm\\
  ( (\texttt{Card}_1(Y^-)\rightarrow \forall y\ Y^-(y)\rightarrow \mathop\lor_{a\in A\atop{r\in Q\atop (r,a,p)\in E}}X_{r,a,p}(y))\land\\
    \hfill(\texttt{Card}_{>1}(Y^-)\rightarrow \forall G'\ \texttt{GMPBMaxEnds}_p(Y^-,G',X))  )\\
\end{multline*}

\noindent Now we define $\psi_{p,q}(X)$ that tests if the connected block $X$ begins in $p$ and ends in $q$.
\begin{multline*}
  \psi_{p,q}(X)\equiv
  \forall^X M^-\forall^X M^+\ \texttt{Min}(M^-,X)\land \texttt{Max}(M^+,X) \rightarrow\\
  ( (\texttt{Card}_1(M^-)\rightarrow \forall m\ M^-(m)\rightarrow \mathop\lor_{a\in A\atop r\in Q} X_{p,a,r}(m))
    \land\\
    \hfill(\texttt{Card}_{>1}(M^-)\rightarrow \forall G'\ \texttt{GMPBMinStarts}_{p}(M^-,G',X)) )\\
  \land\hskip6cm\\
  ( (\texttt{Card}_1(M^+)\rightarrow \forall m\ M^+(m)\rightarrow \mathop\lor_{a\in A\atop r\in Q} X_{r,a,q}(m))
    \land\\
    \hfill(\texttt{Card}_{>1}(M^+)\rightarrow \forall G'\ \texttt{GMPBMaxEnds}_{q}(M^+,G',X)) )
\end{multline*}
Observe that in $\psi_{p,q}(X)$ all quantifications can be assumed relative to $X$. Actually, the formula $\texttt{GMPBMinStarts}_{p}(M^-,G',X)$ ensures that $G'\subseteq X$, and, as $X$ is a connected block and $G'\subset X$, the witnesses of $G'$ are necessarily in $X$.

Let $Q\times Q=\{(p_1,q_1),\dots,(p_n,q_n)\}$.
For every $(p,q)\in Q\times Q$ we define a P-MSO formula $\mathcal{Q}_Y(X,(\psi_{p_1,q_1}(Y),x_1),\dots,(\psi_{p_n,q_n}(Y),x_n),\varphi_{p,q}(x_1,\dots,x_n))$ such that, for any good maximal parallel block $G$ of $P$, 
$$P,G\models \mathcal{Q}_Y(G,(\psi_{p_1,q_1}(Y),x_1),\dots,(\psi_{p_n,q_n}(Y),x_n),\varphi_{p,q}(x_1,\dots,x_n))$$ 
if and only if $G$ can be decomposed into
$G=G_{1,1}\parallel\dots\parallel G_{1,x_1}\parallel\dots\parallel G_{n,1}\parallel\dots\parallel G_{n,x_n}$
such that $p_i \mathop{\Longrightarrow}\limits_{\mathcal{A}}^{G_{j,i}} q_i$ for every $i\in[n]$ and $j\in[x_i]$, and
$\{(p_1,q_1)^{x_1},\dots,(p_n,q_n)^{x_n}\}\in \mathcal{F}_{p,q}$.
In other words,
$P,G\models \mathcal{Q}_Y(G,(\psi_{p_1,q_1}(Y),x_1),\dots,(\psi_{p_n,q_n}(Y),x_n),\varphi_{p,q}(x_1,\dots,x_n))$
if and only if there is a path 
$p \mathop{\Longrightarrow}\limits_{\mathcal{A}}^{G} q$.

As a consequence of Lemma~\ref{lem:FqpParallelRat} and Theorems~\ref{th:ratSemLin} and~\ref{th:PresburgerSemilinear}, $\mathcal{F}_{p,q}$ is the Presburger set (over the alphabet $Q\times Q$) of some Presburger formula $\varphi_{p,q}(x_1,\dots,x_n)$, from which we deduce directly $\mathcal{Q}_Y(X,(\psi_{p_1,q_1}(Y),x_1),\dots,(\psi_{p_n,q_n}(Y),x_n),\varphi_{p,q}(x_1,\dots,x_n))$. 

We now write a formula $\texttt{ConsistentPar}_2$ that applies the adequate $$\mathcal{Q}_Y(G,(\psi_{p_1,q_1}(Y),x_1),\dots,(\psi_{p_n,q_n}(Y),x_n),\varphi_{p,q}(x_1,\dots,x_n))$$ to all good maximal parallel blocks of $P$ (except $P$ itself if it is a good maximal parallel block).
\begin{multline*}
  \texttt{ConsistentPar}_2\equiv \forall X\forall G (X\subset P\land\texttt{GMPB}(G,X))\rightarrow \\
  \forall x\mathop\land_{(p,q)\in Q\times Q}( ((\texttt{Wit}_L(x,G)\land X^-_{p,q}(x))\lor(\texttt{Wit}_R(x,G)\land X^+_{p,q}(x)))\rightarrow\\
  \mathcal{Q}_Y(G,(\psi_{p_1,q_1}(Y),x_1),\dots,(\psi_{p_n,q_n}(Y),x_n),\varphi_{p,q}(x_1,\dots,x_n)) )
\end{multline*}

We are ready to define a P-MSO sentence $\phi_\mathcal{A}$ that defines $L$.
Set $\texttt{ConsistentPar}\equiv\texttt{ConsistentPar}_1\land\texttt{ConsistentPar}_2$.
Observe that for any nonempty N-free poset, either $P=P_1\parallel P_2$ for some nonempty $P_1,P_2$, in which case $P$ is a maximal good parallel block of itself, or $P$ is a connected block. Assuming $A=\{a,b\}$, define
\begin{multline*}
\phi_\mathcal{A}\equiv\exists X_{(p_1,a,q_1)}\exists X_{(p_1,b,q_1)}\exists X^-_{p_1,q_1}\exists X^+_{p_1,q_1}\dots \exists X_{(p_n,a,q_n)}\exists X_{(p_n,b,q_n)}\exists X^-_{p_n,q_n}\exists X^+_{p_n,q_n}\\
\texttt{MarkSeq}\land\texttt{ConsistentSeq}\land\texttt{MarkPar}\land\texttt{ConsistentPar}
\land
( \forall X\ \texttt{Universe}(X)\rightarrow \\
  ((\texttt{Card}_{>0}(X)\rightarrow
    ((\texttt{GMPB}(X,X) \rightarrow\\\hfill \mathop\lor_{i\in I\atop f\in F} \mathcal{Q}_Y(X,(\psi_{p_1,q_1}(Y),x_1),\dots,(\psi_{p_n,q_n}(Y),x_n),\varphi_{i,f}(x_1,\dots,x_n))\\
    \land
    ((\lnot\texttt{GMPB}(X,X)) \rightarrow \mathop\lor_{i\in I\atop f\in F} \psi_{i,f}(X)))
  )\\
  \land
  (\texttt{Card}_{0}(X)\rightarrow \mathop\lor_{i\in I\cap F}\texttt{true}
  ))\hskip4.3cm
)\hskip4cm
\end{multline*}
Then $P\models\phi_\mathcal{A}$ if and only if there is a path $i \mathop{\Longrightarrow}\limits_{\mathcal{A}}^{P} f$ for some $i\in I$, $f\in F$.

\subsection{From P-MSO to automata}

Let $A$ be an alphabet and $\psi(x_1,\dots,x_n,X_1,\dots,X_m)$ be a P-MSO formula which has a set $V_1=\{x_1,\dots,x_n\}$ of free first-order variables interpreted over elements of posets (we do not consider here the variables that are interpreted over non-negative integers) and a set $V_2=\{X_1,\dots,X_m\}$ of free second-order variables. A \emph{$(V_1,V_2)$-poset} labeled by $A$ is a N-free poset $(P,<,\rho)$ labeled by  $A\times {\mathcal P}(V_1)\times {\mathcal P}(V_2)$ such that for all $i\in[n]$ there exists exactly one $p\in P$ such that $x_i\in\pi_2(\rho(p))$, where $\pi_k((c_1,\dots,c_r))=c_k$ ($k\in[r]$) is the projection of a tuple on its $k^\texttt{th}$ component.
Observe that a poset labeled by $A$ can be viewed as a $(\emptyset,\emptyset)$-poset labeled by $A$.
Observe also that an interpretation of the variables $x_{1},\dots,x_{n},X_{1},\dots,X_{m}$ in $P$ induces a unique $(V_1,V_2)$-poset $P(x_{1},\dots,x_{n},X_{1},\dots,X_{m})$, and reciprocally. This allows us to use indifferently one representation or the other in order to lighten the notation.
The $(V_1,V_2)$-posets are a generalization from words to N-free posets of an idea of~\cite{PP86}.

This section is devoted to the construction of an automaton $\mathcal{A}_\psi$ on the alphabet $A\times {\mathcal P}(V_1)\times {\mathcal P}(V_2)$ such that $P,x_1,\dots,x_n,X_1,\dots,X_m\models\psi(x_1,\dots,x_n,X_1,\dots,X_m)$ if and only if $P(x_1,\dots,x_n,X_1,\dots,X_m)\in L(\mathcal{A}_\psi)$ for any $P\in SP^+(A)$. If $\psi$ is a sentence, then $P\models\psi$ if and only if $P\in L(\mathcal{A}_\psi)$. The construction of $\mathcal{A}_\psi$ is by induction on the structure of $\psi$, and is a generalization of the well-known construction of a Kleene automaton from a MSO-formula when MSO is interpreted over words (see for example~\cite{Str94} for a clear presentation of this case).

It is easy to build an automaton $\mathcal{A}_{\text{$(V_1,V_2)$-poset}}$ that tests if a poset $P\in SP^+(A)$ labeled by  $A\times {\mathcal P}(V_1)\times {\mathcal P}(V_2)$ is a $(V_1,V_2)$-poset for some $V_1,V_2$. It suffices to test in $P$, for each $v\in V_1$, if $v$ appears exactly once into the sets that appear as second components of the letters. Example~\ref{ex:oneAAtLeast} exhibits an automaton that tests if a particular letter $a$ appears at least once: it can easily be transformed in order to test if $a$ appears exactly once, from which we deduce $\mathcal{A}_{\text{$(V_1,V_2)$-poset}}$.

As a consequence of Proposition~\ref{prop:union} and Theorem~\ref{th:complement}, from now we consider that inputs of branching automata are $(V_1,V_2)$-posets.

The construction of an automaton $\mathcal{A}_{x_i<x_j}$ that tests if $x_i<x_j$ for some first-order variables in the input $(V_1,V_2)$-poset has  ${\mathcal P}(V_1)\cup\{\bot\}$ as set of states with $\bot$ a sink state. The state $\emptyset$ is the only initial state and all states $V\in {\mathcal P}(V_1)$ such that $x_i,x_j\in V$ are final. 
The sequential transition from state $V\in {\mathcal P}(V_1)$ labeled by $(a,W_1,W_2)$ goes to
\begin{itemize}
\item $\bot$ if  $x_i\not\in V$ and $x_j\in W_1$, or $x_i,x_j\in W_1$;
\item $V\cup W_1$ otherwise.
\end{itemize}
Each state $s$ which is not $\bot$ is the source of a fork transition $(s,\{s,s\})$, and for every pair of states $(W_1,W_2)\in {\mathcal P}(V_1)\times {\mathcal P}(V_1)$ there is a join transition $(\{W_1,W_2\},W_1\cup W_2)$.

As the test automata for the atomic formul\ae\ $a(x_i)$ and $X_i(x_j)$ use the same principle we only give the construction of an automaton $\mathcal{A}_{X_i(x_j)}$ that tests the latter. It has $Q=\{\bot,\top\}$ as set of states with $\bot$ as unique initial state and $\top$ as unique final state.
The only sequential transitions from $\bot$ to $\top$ are labeled by $(a,W_1,W_2)$ such that $x_j\in W_1$ and $X_i\in W_2$. All other sequential transitions are from $s$ to $s$ for all states $s$. The fork transitions are $(s,\{s,s\})$ for all states, and the join transitions are $(\{s_1,s_2\},s_3)$ where $s_3=\top$ if at least one of $s_1,s_2$ is $\top$, $s_3=\bot$ otherwise.

Constructions of automata for the boolean connectors $\lor$, $\land$ and $\lnot$ are a consequence of Proposition~\ref{prop:union} and Theorem~\ref{th:complement}.

Assume now that $\psi(x_1,\dots,x_n,X_1,\dots,X_m)$ is a P-MSO formula with free first-order variables $V_1=\{x_1,\dots,x_n\}$ and free second-order variables $V_2=\{X_1,\dots,X_m\}$.  Assume that by induction hypothesis an automaton $\mathcal{A}_{\psi}=(Q,A\times{\mathcal P}(V_1)\times{\mathcal P}(V_2),E,I,F)$ can effectively be constructed from $\psi(x_1,\dots,x_n,X_1,\dots,X_m)$, and let $i\in[n]$.
We use $\mathcal{A}_{\psi}$ in order to build an automaton $\mathcal{A}_{\exists x_i\psi}=(Q\times\mathbb{B},A\times{\mathcal P}(V_1-\{x_i\})\times{\mathcal P}(V_2),E',I\times\{\false\},F\times\{\true\})$ such that $P,x_1,\dots,x_{i-1},x_{i+1},\dots,x_n,X_1,\dots,X_m\models\exists x_i\psi(x_1,\dots,x_n,X_1,\dots,X_m)$ if and only if $P(x_1,\dots,x_{i-1},x_{i+1},\dots,x_n,X_1,\dots,X_m)\in L(\mathcal{A}_{\exists x_i\psi})$, for any $P\in SP^+(A)$. There is a sequential transition $((q,b),(a,W_1,W_2),(q',b))\in E'$ if and only if $(q,(a,W_1,W_2),q')\in E$ and $x_i\not\in W_1$, and a sequential transition $((q,\false),(a,W_1-\{x_i\},W_2),(q',\true))\in E'$ is and only if $(q,(a,W_1,W_2),q')\in E$ and $x_i\in W_1$. There is a fork transition $((q_1,b),\{(q_2,b),(q_3,b)\})\in E'$ if and only if $(q_1,\{q_2,q_3\})\in E$, and a join transition $(\{(q_1,b_1),(q_2,b_2)\},(q_3,b_1\text{ or } b_2))\in E'$ if and only if $(\{q_1,q_2\},q_3)\in E$.

The  construction for quantification over a second-order variable is similar to the one of first-order variable.

\begin{rem}
  \label{rem:MSOAuto}
  We have proved by all the constructions above that for any MSO-sentence $\psi$ there exists an automaton $\mathcal{A}_\psi$ such that $P\models\psi$ if and only if $P\in L(\mathcal{A}_\psi)$.
\end{rem}

We finally turn to the last case where $\psi$ has the form $$\mathcal{Q}_X(Z,(\psi_1(X),x_1),\dots,(\psi_n(X),x_n),\varphi(x_1,\dots,x_n))$$
Recall here that $x_1,\dots,x_n$ are variables that are interpreted over non-negative integers, and that each $\psi_i$, $i\in[n]$, has one free variable $X$, which is second-order, all quantifications relative to $X$ and no free first-order variables. By induction hypothesis, there is an automaton $\mathcal{A}_{\psi_i}$ such that $P,R\models\psi_i(R)$ if and only if $P,R\in L(\mathcal{A}_{\psi_i})$. According to the semantics of $\mathcal{Q}_X(Z,(\psi_1(X),x_1),\dots,(\psi_n(X),x_n),\varphi(x_1,\dots,x_n))$, the only interpretations of $R$ in $P$ verify (1) $R=P$ and (2) $P$ is a connected block. The conjunction of (1) and (2) is a MSO-definable property of $R$, and thus according to Remark~\ref{rem:MSOAuto} above it can be checked by an automaton $\mathcal{B}$. As a consequence of Proposition~\ref{prop:union} and Theorem~\ref{th:complement} there exists an automaton $\mathcal{A}'_{\psi_i}$ such that $L_i=L(\mathcal{A}'_{\psi_i})=L(\mathcal{A}_{\psi_i})\cap L(\mathcal{B})$. 
Now, let $B=\{b_1,\dots,b_n\}$ be a new alphabet disjoint from $A$. As a consequence of Theorems~\ref{th:PresburgerSemilinear}, \ref{th:ratSemLin} and~\ref{th:KleeneBranching} there is an automaton $\mathcal{C}$ over the alphabet $B$ such that $L(\mathcal{C})$ is the Presburger set of $\varphi(x_1,\dots,x_n)$ over $B$. Then
$L=L_1\circ_{b_1}(\dots(L_n\circ_{b_n}L(\mathcal{C})))$
thus $L$ is regular according to Theorem~\ref{th:KleeneBranching}.

\begin{exa}
  Let $A=\{a,b\}$ and $L\subseteq SP^+(A)$ be the language composed of the nonempty posets of the form $P=P_1\cdot\dots\cdot P_n$, where each $P_i$, $i\in[n]$, has the form $P_i=P_{i,1}\parallel\dots\parallel P_{i,n_i}$ with each $P_{i,j}$ a totally ordered nonempty poset (that is to say, a nonempty word), and such that for each $i\in[n]$ the number of $P_{i,j}$ that starts with an $a$ is ${2\over 3} n_i$.
  Set $L_1=aA^*=a\cup aA^+$ and $L_2=bA^*$.
  Then $L$ is the language of the rational expression $((L_1\parallel L_1\parallel L_2)^\oplus)^+$.
  We define $L$ by a P-MSO sentence as follows.
  Given two elements of the poset denoted by first order variables $x$ and $y$, one can easily write a MSO formula $\texttt{Succ}(x,y)$ (resp. $\texttt{Pred}(x,y)$) that is true if and only if $x$ is a successor (resp. predecessor) of $y$.
  Set
  \begin{align*}
    \texttt{Lin}(X)\equiv& \forall^X x\forall^X y\forall^X z\ (\texttt{Succ}(y,x)\land \texttt{Succ}(z,x)\rightarrow y=z)\land\\&\hskip6cm (\texttt{Pred}(y,x)\land \texttt{Pred}(z,x)\rightarrow y=z)\\
    \psi_1(X)\equiv& \texttt{Lin}(X)\land\exists^X x\ a(x)\land\forall^X y\ x=y\lor x<y\\
    \psi_2(X)\equiv& \texttt{Lin}(X)\land\exists^X x\ b(x)\land\forall^X y\ x=y\lor x<y\\
    \varphi(n_a,n_b)\equiv& n_a=2n_b
    \end{align*}
    Then $L$ is the language of the following P-MSO sentence
    \begin{multline*}
    \psi\equiv \forall P (\forall p\ p\in P)\rightarrow \exists X_1\exists X_2\ P=X_1\oplus X_2\ \\ \land\forall U ((\texttt{MaxBlock}(U,X_1)\lor\texttt{MaxBlock}(U,X_2))\rightarrow \mathcal{Q}_Y(U,(\psi_1(Y),n_a),(\psi_2(Y),n_b),\varphi(n_a,n_b))
  \end{multline*}
  with
  $
    X=U\oplus V\equiv\texttt{Partition}(U,V,X)\land(\forall u\forall v\ u\in U\land v\in V\rightarrow\lnot u\parallel v)
  $.
  In the formula above, $\texttt{Partition}(U,V,X)$ and $u\parallel v$ respectively express with MSO formul{\ae}\ that $(U,V)$ partitions $X$, and that $u$ and $v$ are different and not comparable. The MSO formula $\texttt{MaxBlock}(U,X)$ express that $U$ is a block of $X$, maximal relatively to inclusion.
\end{exa}

\section{Another definition for branching automata}
\label{sec:PBranching}

In this section we introduce another notion of automata for languages of $SP^+(A)$, which is actually equivalent, in expressiveness, to branching automata.

Denote by $\text{Pres}(n)$ the class of Presburger formul\ae\ $\varphi(x_1,\dots,x_n)$ with $n$ free variables.
A \emph{Presburger-branching automaton} (or \emph{P-branching automaton} for short) is a tuple $\mathcal{A}=(Q, A, E, I, F)$ where $Q=\{q_1,\dots,q_n\}$ is a finite set of
states, $A$ is an alphabet, $I\subseteq Q$ is the set of \emph{initial states},
$F\subseteq Q$ the set of \emph{final states}, and $E=(E_\text{seq}, E_\text{fork}, E_\text{join})$ is the set of \emph{transitions} of ${\mathcal A}$, which are of three kinds:
\begin{itemize}
\item $E_\text{seq}\subseteq (Q\times A\times Q)$ contains the \emph{sequential} transitions, which are usual transitions of Kleene automata;
\item $E_\text{fork}$ and $E_\text{join}$, respectively the \emph{fork} and \emph{join} transitions, are finite subsets of $Q\times \text{Pres}(n)$.
\end{itemize}

Let $F$ be a multi-set of elements of $Q$. By $\mu(F)$ we denote the Parikh's image of $F$, ie. the element $(v_1,\dots,v_n)$ of $\mathbb{N}^n$ such that $F$ is the multi-set with $v_i$ occurrences of $q_i$, for all $i\in[n]$. 
Paths in P-branching automata are defined as in branching automata, except for the parallel case: 
for any finite set of paths $\{\gamma_0,\dots,\gamma_k\}$ (with $k\geq 1$) respectively labeled by $P_0,\dots,P_k\in SP^+(A)$, from $r_0,\dots,r_k$ to $s_0,\dots, s_k$, if $t=(p,\varphi(x_1,\dots,x_{n}))$ is a fork transition, $t'=(q,\varphi'(x'_1,\dots,x'_n))$ a join transition, then $\gamma=t(\parallel_{j\leq k} \gamma_j)t'$ is a path from $p$ to $q$ and labeled by $\parallel_{j\leq k} P_j$ if $\mu(\{r_0,\dots,r_k\})=(v_1,\dots,v_n)$ and $\mu(\{s_0,\dots,s_k\})=(v'_1,\dots,v'_n)$ are respectively in the Presburger set of $\varphi$ and of $\varphi'$, and $\sum_{i\in[n]}v_i$, $\sum_{i\in[n]}v'_i>1$.

The notions of accepting paths and languages of a P-branching automaton are defined as in branching automata. A language $L\subseteq SP^+(A)$ is \emph{P-regular} if there exists a P-branching automaton $\mathcal{A}$ such that $L=L(\mathcal{A})$.

\begin{thm}
  Let $A$ be an alphabet and $L\subseteq SP^+(A)$.
  Then $L$ is regular if and only if it is P-regular.
\end{thm}

\begin{Proof}
  It is immediate that regularity implies P-regularity, since a fork transition (the same applies for join transitions) $(p,\{r_0,\dots r_k\})$ of a branching automaton can be interpreted as a fork transition $(p,\varphi(x_1,\dots,x_n))$ of a P-branching automaton, with $\mu(\{r_0,\dots r_k\})$ the unique element of the Presburger set of $\varphi(x_1,\dots,x_n)$.

The converse needs more attention since the Presburger set of a formula $\varphi$ may be infinite. Assume that $\mathcal{A}$ is a P-branching automaton with $Q=\{q_1,\dots,q_n\}$ as set of states. We replace one by one all of its fork transitions $(p,\varphi(x_1,\dots,x_n))$ by fork transitions of branching automata, by the following construction that add new states to $\mathcal{A}$. Denote by $L_\varphi=\{ F\in\mathcal{M}^{>1}(Q) : \mu(F)\text{ is in the Presburger set of }\varphi\}$. Each element of $L_\varphi$ can equivalently be represented as an element of $Q^\circledast$. We first build a branching automaton $\mathcal{A}_\varphi$ on the alphabet $Q$ such that $L(\mathcal{A}_\varphi)=L_\varphi$. It has exactly one initial state, which is not the destination of any transition, and one final state which is not the source of any transition. By Theorem~\ref{th:PresburgerSemilinear}, $L_\varphi$ is a semi-linear set of $Q^\circledast$:
$$
L_\varphi=\bigcup_{i\in I}q_{j_{i,1}}\parallel\dots\parallel q_{j_{i,l_i}}\parallel(\bigcup_{k\in K_i} q_{j_{k,1}}\parallel\dots\parallel q_{j_{k,l_k}})^\circledast
$$
for some finite set $I$, disjoint finite sets $K_i$, $i\in I$, with $j_{i,r}, j_{k,s}\in[n]$ for all $i\in I, k\in K_i$, $r\in[l_i]$, $s\in[l_k]$. Wlog. we may assume that each multi-set of $L_\varphi$ has at least two elements, so $l_i>1$ for all $i\in I$. Then $\mathcal{A}_\varphi$ is composed of one initial state $1$, one final state $f$, two states $\overline{q_{j_{i,r}}}$ and $\overline{\overline{q_{j_{i,r}}}}$ and one sequential transition $(\overline{q_{j_{i,r}}},q_{j_{i,r}},\overline{\overline{q_{j_{i,r}}}})$ for all $i\in I$, $r\in[l_i]$, two states $\overline{q_{j_{k,r}}}$ and $\overline{\overline{q_{j_{k,r}}}}$ and one sequential transition $(\overline{q_{j_{k,r}}},q_{j_{k,r}},\overline{\overline{q_{j_{k,r}}}})$ for all $i\in I$, $k\in K_i$, $r\in[l_k]$. For each $i\in I$, there is one fork transition from $1$ to all the $\overline{q_{j_{i,r}}}$, $r\in[l_i]$, and symmetrically, one join transition from all the $\overline{\overline{q_{j_{i,r}}}}$, $r\in[l_i]$, to $f$. For each $i\in I$, add also two states $\overline{u_i}$ and $\overline{\overline{u_i}}$, a fork transition $(\overline{u_i},\{\overline{u_i},\overline{u_i}\})$, a join transition $(\{\overline{\overline{u_i}},\overline{\overline{u_i}}\},\overline{\overline{u_i}})$, and, for all $k\in K_i$, a fork transition $(\overline{u_i},\{\overline{q_{j_{k,1}}},\dots,\overline{q_{j_{k,l_k}}}\})$ and a join transition $(\{\overline{\overline{q_{j_{k,1}}}},\dots,\overline{\overline{q_{j_{k,l_k}}}}\},\overline{\overline{u_i}})$ if $l_k>1$, a sequential transition $(\overline{u_i},q_{j_{k,1}},\overline{\overline{u_i}})$ if $l_k=1$. For each $i\in I$, add a fork transition $(1,\{\overline{q_{j_{i,1}}},\dots,\overline{q_{j_{i,l_i}}},\overline{u_i}\})$ and a join transition $(\{\overline{\overline{q_{j_{i,1}}}},\dots,\overline{\overline{q_{j_{i,l_i}}}},\overline{\overline{u_i}}\},f)$. We have $L(\mathcal{A}_\varphi)=L_\varphi$. Now, remove all sequential transitions, all join transitions, $f$ and all states of the form $\overline{\overline{q}}$ from $\mathcal{A}_\varphi$ and name $\mathcal{A}'_\varphi$ the new automaton. Consider the disjoint union of $\mathcal{A}$ and $\mathcal{A}'_\varphi$. Remove $(p,\varphi(x_1,\dots,x_n))$. Replace $1$ by $p$ and all states of the form $\overline{q}$, $q\in Q$, by $q$. 
Join transitions $(p,\varphi(x_1,\dots,x_n))$ are removed with a similar mechanism.
\end{Proof}

\section{Conclusion}
\label{sec:conclusion}

The effectiveness of the constructions involved in the proof of Theorem~\ref{th:PMSORat} have several consequences.
The \emph{P-MSO theory $S$ of $SP^+(A)$} consists of all sentences $\phi$ of P-MSO such that $P\models\phi$ for every $P\in SP^+(A)$. % cf Ros82 p233
The P-MSO theory of $SP^+(A)$ is \emph{decidable} if there exists a decision procedure that tests if $\phi\in S$.
Because emptiness is decidable for languages of branching automata (see Proposition~\ref{prop:findPath}), we have:
\begin{thm}
  \label{th:PMSODecidable}
  Let $A$ be an alphabet.
  The P-MSO theory of $SP^+(A)$ is decidable.
\end{thm}

In~\cite{LW00:sp}, Lodaya and Weil asked for logical characterizations of several classes of rational languages. 
As it is equivalent to branching automata, P-MSO is the natural logic to investigate such questions, that are still open.

Extending the work of Lodaya and Weil, and those of Kleene-Sch\"utzenberger, Kuske and Meinecke~\cite{Kuske200453} proposed to attach costs to paths in branching automata. They defined and studied branching automata with costs, and extended to this case the machinery from the theory of weighted automata. They provided in particular a Kleene-like theorem for branching automata with costs, in the particular case of bounded-width languages.

Among the works connected to ours, let us mention \'Esik and N{\'e}meth~\cite{EN01}, which itself has been influenced by the work of Hoogeboom and ten Pas~\cite{HtP:1996,HtP:1997} on text languages. They study languages of bi-posets from an algebraic, automata and regular expressions based point of view, and the connections with MSO. A bi-poset is a set equipped with two partial orderings; thus, N-free posets are a generalization of N-free bi-posets, where commutation is allowed in the parallel composition. 

MSO and Presburger logic were also mixed in other works, but for languages of trees instead of N-free posets. Motivated by reasoning about XML documents, Dal Zilio and Lugiez~\cite{DBLP:conf/rta/Dal-ZilioL03}, and independently Seidl, Schwentick and Muscholl~\cite{DBLP:conf/birthday/SeidlSM08}, defined a notion of tree automata which combines regularity and Presburger arithmetic. In particular in~\cite{DBLP:conf/birthday/SeidlSM08}, MSO is enriched with Presburger conditions on the children of nodes in order to select XML documents, and proved equivalent to unranked tree automata. Observe that unranked trees are a particular case of N-free posets.

\section*{Acknowledgement}

The author would like to thank the referees of this paper, whose comments helped in improving its quality. Among many remarks, the content of Section~\ref{sec:PBranching} was suggested by one of them.

%%%%%%%%%%%%%%%%%%%%%%%%%%%%%%%%%%%%%%%%%%%%%%%%%%%%%%%%%%%%%%%%%%%%%%%%%%%%%%%
%\bibliographystyle{alpha} 
%\bibliography{../Biblio/allInfiniteWords}

%%%%%%%%%%%%%%%%%%%%%%%%%%%%%%%%%%%%%%%%%%%%%%%%%%%%%%%%%%%%%%%%%%%%%%%%%%%%%%%

\end{document}